\documentclass[11pt]{article}
\usepackage[utf8]{inputenc}
\input{preamble.sty}

\title{Optimal Contextual Pricing and Extensions}
\author{Allen Liu\thanks{Massachusetts Institute of Technology, email: cliu568@gmail.com} \and Renato Paes Leme\thanks{Google Research, email: renatoppl@google.com} \and Jon Schneider\thanks{Google Research, email: jschnei@google.com}}
\date{}

\begin{document}

\maketitle

\begin{abstract}
In the contextual pricing problem a seller repeatedly obtains products described by an adversarially chosen feature vector in $\R^d$ and only observes the purchasing decisions of a buyer with a fixed but unknown linear valuation over the products. The regret measures the difference between the revenue the seller could have obtained knowing the buyer valuation and what can be obtained by the learning algorithm.

We give a poly-time algorithm for contextual pricing with $O(d \log \log T + d \log d)$ regret which matches the $\Omega(d \log \log T)$ lower bound up to the $d \log d$ additive factor. If we replace pricing loss by the symmetric loss, we obtain an algorithm with nearly optimal regret of $O(d \log d)$ matching the $\Omega(d)$ lower bound up to $\log d$. These algorithms are based on a novel technique of bounding the value of the Steiner polynomial of a convex region at various scales. The Steiner polynomial is a degree $d$ polynomial with intrinsic volumes as the coefficients.

We also study a generalized version of contextual search where the hidden linear function over the Euclidean space is replaced by a hidden function $f : \X \rightarrow \Y$ in a certain hypothesis class $\H$. We provide a generic algorithm with $O(d^2)$ regret where $d$ is the covering dimension of this class. This leads in particular to a $\tilde{O}(s^2)$ regret algorithm for linear contextual search if the linear function is guaranteed to be $s$-sparse. Finally we also extend our results to the noisy feedback model, where each round our feedback is flipped with a fixed probability $p < 1/2$.
\end{abstract}

\pagebreak

\section{Introduction}

In the \emph{contextual search problem} a learner tries to learn a hidden linear function $x \in \R^d \mapsto \dot{v}{x}$ for some unknown $v \in \R^d$
. In every round, the learner is presented with an adversarially chosen vector $x_t \in \R^d$, 
and is asked to provide a guess $y_t \in \R$ for the dot-product $\dot{v}{x_t}$, subsequently learning whether $y_t \leq \dot{v}{x_t}$ or $y_t > \dot{v}{x_t}$ and incurring a loss $\ell(y_t, \dot{v}{x_t})$. The goal of the learner is to minimize the total loss (the \textit{regret}), which is given by $\sum_t \ell(y_t, \dot{v}{x_t})$. 

A special case of this problem is \emph{contextual pricing} \cite{kleinberg2003value,amin2014repeated,CohenLL16,nazerzadeh2016,lobel2016multidimensional,qiang2016dynamic,Intrinsic18, krishnamurthy2020corrupted}. In this setup, the vectors $x_t$ are features representing differentiated products and the learner is a seller whose decision at round $t$ is how to price item $x_t$. Given a price $y_t$ a buyer with valuation $u_t = \dot{v}{x_t}$ buys the product if $u_t \geq y_t$ and doesn't buy otherwise. The seller only observes the purchase or no-purchase decision. The loss in each round is the difference between the revenue made by the seller $y_t \cdot \one\{ y_t \leq u_t \}$ and the revenue the seller could have made if $v$ was known. Formally, the pricing loss is given by:
$$\ell(y_t, u_t) = u_t - y_t \cdot \one\{ y_t \leq u_t \}$$

A second important case is called \emph{symmetric contextual search} where the loss is the difference between the guess $y_t$ and the actual dot product $\dot{v}{x_t}$. This loss function arises in \emph{personalized medicine} \cite{bayati2016} where the learner chooses the dosage of a medicine and observes whether the patient was over-dosed or under-dosed. Another application is one-bit compressed sensing \cite{plan2012robust,davenport20141} where the learner only observes the sign of a measurement. In either case, we will consider the following loss:
$$\ell(y_t, u_t) = \abs{ y_t -  u_t}$$

\paragraph{Optimal Regret Bounds for Contextual Search} 
In one dimension, both contextual pricing and symmetric contextual search reduce to non-contextual problems and are well understood. For the symmetric loss the optimal regret in the one-dimensional case is $\Theta(1)$ using binary search. For pricing, the optimal regret in the one-dimensional case is $\Theta(\log \log T)$ using the algorithm of Kleinberg and Leighton \cite{kleinberg2003value}. These results immediately imply $\Omega(d)$ and $\Omega(d \log \log T)$ lower bounds for the general contextual case (see Section \ref{sec:optimality} for details on optimality). 

In this paper, we design polynomial-time algorithms with regret $O(d \log \log T + d\log d)$ for contextual pricing and $O(d \log d)$ for symmetric contextual search, matching these lower bounds up to a $\log d$ factor.  This improves over the previously known bounds in \cite{Intrinsic18}, which are $O(d^4 \log \log T)$ for pricing and $O(d^4)$ for symmetric contextual search.

\paragraph{Steiner polynomial} The main technique driving these results is a new potential function based on the \emph{Steiner polynomial}. This is an object from integral geometry that is closely connected with the notion of \emph{intrinsic volumes} which were used in \cite{Intrinsic18} to derive the previously known bounds for this problem.

Given a convex set $S \subseteq \R^d$, Steiner showed that the volume of the Minkowski sum $\Vol(S + z \B)$ is a polynomial of degree $d$ in $t$ (where $\B$ is the unit ball). The intrinsic volumes $V_j$ of $S$ correspond to the coefficients of this polynomial after normalization by volume $\kappa_{j}$ of the $j$-dimensional ball:
$$\Vol(S + z \B) = \sum_{j=0}^d  V_j(S) \kappa_{d-j} z^{d-j}$$

Both our algorithm and the the algorithm in \cite{Intrinsic18} keep track of the set $S_t$ of vectors consistent with observations seen so far. The intrinsic volumes approach keeps track of $V_j(S_t)$ and shows that the loss incurred in round $t$ is proportional to the decrease of one of the suitably normalized intrinsic volumes (i.e. $V_j(S_j)^{1/j}$ for some index $j \in \{1, \hdots, d\}$). In this paper, instead of keeping track of each coefficient individually, we control the value of the Steiner polynomial itself at different values of $z$. Specifically, we show that for some set of $d$ values $\{z_1, z_2, \dots, z_d\}$ it is possible to always choose an $i$ (based on the current width of our set) so that $\Vol(S + z_i \B)$ decreases by a constant fraction. This leads to nearly optimal bounds on regret (via a much simpler proof than that in \cite{Intrinsic18}).


\paragraph{Framework for learning with binary feedback} 

While the Steiner polynomial technique largely resolves the classical problem of contextual search, there is a wide class of learning problems with binary feedback that either do not fit within the framework of learning a linear function, or which impose additional constraints on the linear function that one would hope to leverage.

One example is \emph{sparse contextual search}, where the hidden vector $v \in \R^d$ is guaranteed to be $s$-sparse, i.e., $\norm{v}_0 \leq s$. This captures settings where we expect few features to matter to the buyer. Another interesting problem is when we are asked to guess $\max_i v_i x_i$ instead of $\dot{v}{x}$. This corresponds to learning the valuation of an \emph{unit-demand buyer}. This problem is challenging since the set of vectors $v \in \R^d$ consistent with observations seen so far is not necessarily convex.

Both of these examples are special cases of a general framework for online learning problems under binary feedback. In our general setup, the learner is trying to learn a function $f$ in a hypothesis class $\H$ containing functions mapping from a context space $\X$ to an outcome space $\Y$. In each step a context $x_t \in \X$ is chosen adversarially and the learner is asked to submit a guess $y_t$ for the value of $f(x_t)$ and incurs a loss $\ell(y_t, f(x_t))$. The goal of the learner is to minimize the total loss $\sum_{t=1}^T \ell(y_t, f(x_t))$. The original setup corresponds to the case where  $\X$ is some subset of $\R^d$ e.g. the unit ball $\B \subseteq \R^d$, $\Y = [-1,1]$, $\H$ is the class of all linear functions $f_v(x) = \dot{v}{x}$ for $v \in \B$.

Our main result in this space is an algorithm with regret $O(d^2)$ where $d$ is the covering dimension of the hypothesis class $\H$ (see Definition \ref{def:covering_dim}). This result immediately improves the regret of symmetric contextual search from $\tilde{O}(d)$ to $\tilde{O}(s^2)$ \footnote{We write $\tilde{O}$ to suppress any logarithmic factors in $d$ or $s$.}. Similarly, this result immediately implies an $O(d^2)$ regret algorithm for the unit-demand buyer problem. We accomplish this by generalizing the Steiner polynomial idea for linear contextual search to a general ``Steiner potential'' defined for any hypothesis class (see the Techniques subsection below). 

We contrast these results in Section \ref{sec:full_feedback} with the full feedback case in which the algorithm learns $f(x_t)$ each round. In this full-feedback setting, we give matching upper and lower bounds (up to constant factors) on the achievable regret. Our results here are based on a notion we introduce of \textit{tree-dimension} of a hypothesis class, which is a continuous analogue of Littlestone dimension.

\paragraph{Techniques in the general case} The Steiner polynomial is defined for convex sets living in the Euclidean space.
Intriguingly, it is possible to generalize (in some sense) this geometric technique to arbitrary classes $\H$ of hypotheses. Instead of keeping track of all functions in the hypothesis class that are consistent with the feedback so far, we keep track of an expanded set of functions that don't violate the feedback up to a certain margin (in the linear case, this is exactly $S + \lambda\B$). This has the effect of regularizing the set of consistent hypotheses and allows for faster progress. Instead of volume, in the general case we control the size of an $\epsilon$-net of the set of these approximately valid hypotheses. 

A second technique we use is \emph{adaptive scaling}, which involves keeping track of multiple levels of discretization. For the linear case, this boils down to controlling the value of the Steiner polynomial at different values of $\lambda$. More generally, at each step, we can estimate the maximum possible loss achievable in this round given the previous feedback. Based on this value, we will choose a scale, which will dictate the granularity of the $\epsilon$-net and the margin with which we prune inconsistent hypotheses. After picking the scale we show that it is possible to pick a (random) cut that will either: (i) reduce the number of valid hypotheses in the chosen granularity by half; or (ii) eliminate one valid hypothesis at a much coarser granularity. This will require a careful coupling between the discretizations at two different levels. This coupling between two levels is what allows us to overcome the fact that in the general case we can't rely on techniques from convex geometry. See Section \ref{sec:intuition_const_T} for details.

One important feature of all our algorithms (not shared by previous algorithms) will be our use of randomness, in particular \emph{perturbed guesses}. Every round, we compute the median $m_t$ of the set $f(x_t)$ where $f$ ranges over the set of approximately valid hypotheses. However, instead of guessing the median $m_t$ directly we guess one of the two values (chosen uniformly at random) in $\{m_t - \delta, m_t + \delta\}$, where the size of perturbation $\delta$ depends on our current scale $\lambda$. Our guarantee is that the potential function will decrease significantly for one of the two choices (and thus in expectation).

\paragraph{Noisy Contextual Search} The final direction in which we extend the original contextual search problem is by considering noisy binary feedback, i.e., the feedback of the algorithm is flipped with probability $p < 1/2$.  In this setting we move from keeping track of a set of approximately valid hypotheses to a pseudo-Bayesian approach, where we maintain a distribution $w$ over approximately valid hypotheses and update it as we receive feedback. By carefully bounding the weight of hypotheses within a ball of radius $1/T$, this results in an algorithm with regret $O(d\log T)$. 


Ideally, it would be possible to combine this algorithm with the adaptive scaling technique of the noiseless setting, resulting in an $O(\poly(d))$ regret algorithm for general hypothesis classes. One such approach is to replace the notion of width with a fuzzier notion, based on how tightly concentrated the distribution is along the current context vector (e.g. the width of the smallest strip in this direction which contains $1-\eps$ of the mass of the distribution). We can then choose a scale based on this distributional width, and choose the size of the perturbation based on this scale (as in the deterministic case). 

This type of approach works, conditional on being able to show that when the distribution concentrates along a thin strip, the true hypothesis is close to this thin strip with high probability. Unfortunately, doing this for general hypothesis classes seems hard -- fortunately, it is possible to do this for the specific case of symmetric contextual search by leveraging the Euclidean geometry of the ambient space (see Section \ref{subsec:polyd_noisy} for more details). This leads to an algorithm for noisy linear contextual search which gets $O(\poly(d))$ regret, and is the first algorithm we are aware of for contextual search in the noisy setting which gets any regret independent of $T$ for $d > 1$.

\paragraph{Summary of main results} To summarize, our results include:

\begin{itemize}
    \item Algorithms with regret $O(d \log d)$ for symmetric contextual search (Section \ref{subsec:dlogd}) and $O(d \log \log T + d \log d)$ for contextual pricing (Section \ref{subsec:dloglogt}). Both algorithms are optimal (up to $\log d$) and have only $O(d)$ overhead with respect to the non-contextual case. Both algorithms can be implemented efficiently in $\poly(d, T)$ time.
    \item General algorithm for learning a function from a hypothesis class $\H$ under binary feedback with regret $O(d^2)$ where $d$ is the covering dimension of the hypothesis class (Section \ref{subsec:d_square_noiseless}).
    \item An algorithm for symmetric contextual search with noisy binary feedback with $O(\poly(d))$ regret (Section \ref{subsec:polyd_noisy}).
\end{itemize}



\paragraph{Related work}

Core to our results is the idea of coupling together potentials at many different scales. Similar ideas of ``adaptive discretization'', ``zooming'', and ``chaining'' exist throughout the online learning literature \cite{bubeck2011x, kleinberg2019bandits, slivkins2014contextual} and statistical learning theory literature \cite{gaillard2015chaining,cesa2017algorithmic}. Algorithms in these works also often construct several layers of discretizations and have learning rates parameterized by the covering dimension of the ambient space. However, these algorithms are usually designed for settings where (1) one cannot hope for better than $O(\sqrt{T})$ regret (let alone regret independent of $T$), (2) feedback is not binary but rather zeroth-order (\cite{mao2018contextual} study a pricing setting where feedback is binary, but where the hypothesis class is large enough that one must incur $O(\poly(T))$ regret). In particular, we believe our technique of coupling together the potentials for different scales in the analysis of Theorem \ref{main_binaryfeedback} is novel.

Our results in the full feedback case (Section \ref{sec:full_feedback}) -- parameterizing the optimal regret in terms of the tree dimension -- can be seen as a generalization of similar results for Littlestone dimension \cite{littlestone1988learning} (indeed, in the case where $\Y = \{0, 1\}$, our notion of tree dimension reduces to Littlestone dimension). While there do exist measures which capture the learnability of functions taking values over a metric space (for example, the fat-shattering dimension \cite{bartlett1996fat} for real-valued functions), as far as we are aware the notion of tree dimension we introduce does not currently exist in the literature. It is an interesting open direction to connect the notion of tree dimension we present with previously studied measures. 

\section{Optimal Contextual Search}

We start by describing the contextual search setup and establishing some useful notation. The hidden object is a vector $v_0$ belonging to the unit ball $\B = \{v \in \R^d; \norm{v}_2 \leq 1\}$. In each round $t \in \{1, \dots, T\}$ the learner is provided an (adversarially chosen) vector $x_t \in \B$ and asked to provide a guess $y_t \in \R$ \footnote{Setting the domain for $v$ and $x_t$ to be the unit ball will simplify the analysis.  For the optimality construction, we will instead assume that $v$ is drawn from the $L^{\infty}$ ball and that $x_t$ is drawn from the $L^1$ ball.  It will be straight-forward to extend our analysis to this case.  See Section \ref{sec:optimality} for details. }.  Upon guessing, the learner incurs loss $\ell(y_t, \dot{v_0}{x_t})$ and receives feedback $\sigma_t \in \{-1, +1\}$ corresponding to whether  $y_t > \dot{v_0}{x_t}$ ($\sigma_t = +1$) or $y_t < \dot{v_0}{x_t}$ ($\sigma_t = -1$). If $y_t = \dot{v_0}{x_t}$, then the feedback is arbitrary. In other words:
$$\sigma_t = \sign(y_t - \dot{v_0}{x_t})$$
This allows the learner to keep track of the set $S_t$ of vectors consistent with observations seen so far:
$$S_t := \{v \in \B; \sigma_{\tau}(y_\tau - \dot{v}{x_\tau}{}) \geq 0 \text{ for all } \tau < t\}$$
It is clear from the above setup that for both pricing and symmetric loss, it suffices to consider when $||x_t||_2 = 1$ for all rounds.

Throughout the execution of the algorithm we will keep track of the Steiner potential $\Vol(S_t + z\B)$ where  $\Vol(\cdot)$ is the standard volume in $\R^d$ and the sum is the Minkowski sum:
$$S_t + z\B = \{w + z\cdot u; w \in S_t \text{ and } u \in \B\} $$

We will evaluate the potential at different points depending on the width of $S_t$ in the direction $x_t$. We define the width as:

$$\width(S_t; x_t) = \max_{v \in S_t} \dot{v}{x_t} - \min_{v \in S_t} \dot{v}{x_t}$$

\subsection{Symmetric loss with $O(d \log d)$ regret}\label{subsec:dlogd}

We start with the symmetric loss function $\ell(y_t, u_t) = \abs{y_t - u_t}$, where we will show it is possible to obtain $O(d\log d)$ regret. The main idea of this algorithm is to choose a value $z_i$ based on the width of $S_t$ in the direction of the current context and then choose a guess $y_t$ that splits the set $S_t + z_i \B$ in two parts of equal volume. By doing this, we will show that $\Vol(S_t + z_i \B)$ (the ``Steiner potential'') decreases by a constant multiplicative fraction. Since $\Vol(S_t + z_i \B)$ is bounded below by $\Vol(z_i\B)$, we can only do this some number of times (roughly $d\log (1/z_i)$ times), from which our regret bound will follows.

We describe the algorithm below (ignore for now issues of computational efficiency; we will address these in Section \ref{sec:efficiency}):

\begin{algorithm}[H]
\caption{{\sc Multiscale Steiner Potential}}
\label{alg:linear_symmetric}
\begin{algorithmic} 
\State Initialize $S_1 = \B$ and $z_i = 2^{-i} / (8d)$  for all $i$ 
\For {$t$ in $1,2, \dots , T$}
\State Adversary picks $x_t$
\State Let $i$ be the largest index such that $\width(S_t; x_t) \leq 2^{-i}$
\State Query $y_t$ such that $\Vol(\{v \in S_t + z_i \B; \dot{v}{x_t} \geq y_t\}) = \frac{1}{2} \Vol(S_t + z_i \B)$
\EndFor

\end{algorithmic}
\end{algorithm}

\begin{lemma}\label{lemma:volume_progress_symmetric}
If $i$ is the index chosen at time $t$ then $\Vol(S_{t+1} + z_i \B) \leq \frac{3}{4} \Vol(S_t + z_i \B)$.
\end{lemma}

\begin{proof}
Assume the feedback is $\sigma_t = +1$ (the other case is analogous). Then $S_{t+1} = \{v \in S_t; \dot{v}{x_t} \geq y_t \}$. In Figure  \ref{fig:dlogd} we depict the set $S_t + z_i \B$ and $S_{t+1} + z_i \B$. The part of $S_{t+1} + z_i \B$ with $\dot{v}{x_t} \geq y_t$ is exactly $\{v \in S_t + z_i \B; \dot{v}{x_t} \geq y_t\}$ which has volume $\frac{1}{2} \Vol(S_t + z_i \B)$.

To bound the remaining part, let $C$ be the largest volume of a section of $S_t + z_i \B$ in the direction $x_t$ (see the right part of Figure \ref{fig:dlogd}). The total volume of  $S_t + z_i \B$ can be bounded below by comparing it to the two cones formed by taking the convex hull of the section of largest volume inside the band and the extreme points $q_1$ and $q_2$, which are at least $2^{-(i+1)}$ apart in the $x_t$ direction. The volume of the two cones is at least:
$$\Vol(S_t + z_i \B) \geq \frac{2^{-(i+1)} C}{d}  =\frac{C 2^{-i}}{2d} = \frac{C \cdot 8 z_i d}{2d} = 4  z_i \cdot C $$ 

Finally, note that the region of $S_{t+1} + z_i \B$ with $\dot{v}{x_t} \leq y_t$ has cross-section with volume at most $C$ and width $z_i$ in the $x_t$ direction so its volume is at most $C z_i \leq \frac{1}{4} \Vol(S_t + z_i \B) $, thus completing the proof.\end{proof}

\begin{theorem}\label{thm:symloss}
The  regret of the {\sc Multiscale Steiner Potential} algorithm is at most $O(d \log d)$.
\end{theorem}

\begin{proof}
Every time we choose index $i$, the loss at is most $2^{-i}$ and the volume of $\Vol(S_{t} + z_i \B)$ decreases by a constant factor. The set $S_t$ is never empty since $v_0 \in S_t$ for all $t$, therefore $\Vol(S_{t} + z_i \B) \geq \Vol(z_i \B) = z_i^d \Vol(B)$. For this reason we can't pick index $i$ by more than $O(d \log(1/z_i))$ times, so the total regret is at most:
$$O \left(\sum_i 2^{-i} d \log(1/z_i)\right) = O \left(\sum_i d 2^{-i} (i + \log(d)) \right) = O(d \log d)$$
\end{proof}

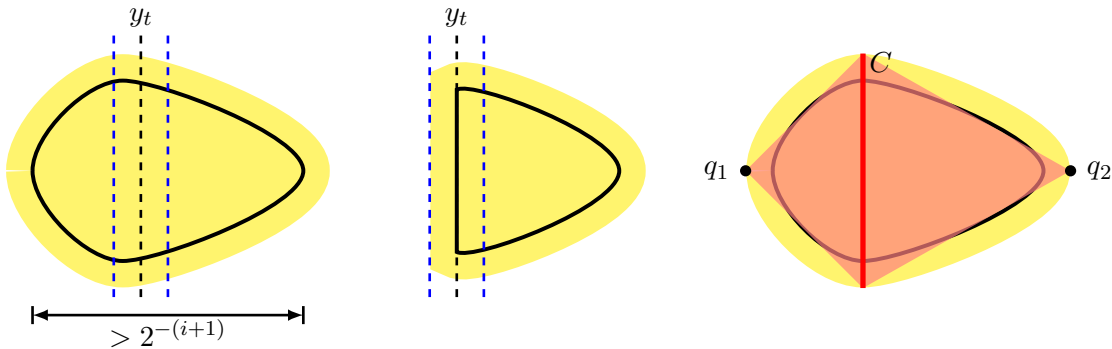
\begin{figure}[h]
\centering
\begin{tikzpicture}[scale=1.2]
  \newcommand{\pathA}{(0,1) .. controls (0,1.4) and (.6, 2) .. (1,2)
              .. controls (1.4,2) and (3,1.4) .. (3,1)
              .. controls (3,0.6) and (1.4,0) .. (1,0)
              .. controls (.6,0) and (0,.6) .. (0,1)}
  \fill[blue!0!white] (-.1,-.4) rectangle (3.4,2.6);
  \draw[line width=20pt, color=yellow!70!white] \pathA;
  \draw[line width=1.5pt, fill=yellow!70!white] \pathA;
 \draw[line width=1pt, dashed] (1.2,-.4)--(1.2,2.5);
 \draw[line width=1pt, dashed, color=blue] (0.9,-.4)--(0.9,2.5);
 \draw[line width=1pt, dashed, color=blue] (1.5,-.4)--(1.5,2.5);
 \node at (1.2, 2.7) {$y_t$};
 \node at (1.5, -.8) {$>2^{-(i+1)}$};

 \begin{scope}[line width=1.0pt]
 \begin{scope}[>=latex]
 \draw[<->] (0,-.6)--(3,-.6);
 \end{scope}
 \draw (0,-.5)--(0,-.7);
 \draw (3,-.5)--(3,-.7);
 \end{scope}

 \begin{scope}[xshift=3.5cm]

  \newcommand{\pathC}{(1.2,1.9) .. controls (1.4,2) and (3,1.4) .. (3,1)
              .. controls (3,0.6) and (1.4,0) .. (1.2,0.1) -- cycle}
  \fill[blue!0!white] (-.1,-.4) rectangle (3.4,2.6);
  \draw[line width=20pt, color=yellow!70!white] \pathC;
  \draw[line width=1.5pt, fill=yellow!70!white] \pathC;
 \draw[line width=1pt, dashed] (1.2,-.4)--(1.2,2.5);
 \draw[line width=1pt, dashed, color=blue] (0.9,-.4)--(0.9,2.5);
 \draw[line width=1pt, dashed, color=blue] (1.5,-.4)--(1.5,2.5);
 \node at (1.2, 2.7) {$y_t$};
 \end{scope}
 \begin{scope}[xshift=8.2cm]
  \renewcommand{\pathA}{(0,1) .. controls (0,1.4) and (.6, 2) .. (1,2)
              .. controls (1.4,2) and (3,1.4) .. (3,1)
              .. controls (3,0.6) and (1.4,0) .. (1,0)
              .. controls (.6,0) and (0,.6) .. (0,1)}
  \draw[line width=20pt, color=yellow!70!white] \pathA;
  \draw[line width=1.5pt, fill=yellow!70!white] \pathA;
  \fill[line width=0pt, fill=red!50!white, fill opacity=0.7] (-.3,1) -- (1,2.3) -- (3.3,1) -- (1,-.3) -- cycle;
  
 \node [shape=circle, fill=black,inner sep=1.5pt,label=left:$q_1$] (0,1) at (-.3,1) {};
 \node [shape=circle, fill=black,inner sep=1.5pt,label=right:$q_2$] (3,1) at (3.3,1) {};
 \node at (1.2, 2.2) {$C$};

 \draw[line width=2pt, color=red] (1,2.3)--(1,-.3);
\end{scope}
\end{tikzpicture}
\caption{Illustration of the proof of Lemma \ref{lemma:volume_progress_symmetric}}
\label{fig:dlogd}
\end{figure}

\subsubsection{Comparison with other approaches}
Is the Steiner potential necessary? One natural algorithm for this problem is to query $y_t$ such that $\Vol(\{v \in S_t; \dot{v}{x_t} \geq y_t\}) = \frac{1}{2} \Vol(S_t)$ (i.e. guess the median without inflating the set). The best upper bound from \cite{Intrinsic18} shows only that this has regret at most $2^{O(d \log d)}$ and a lower bound given in the example in Section 8 of \cite{lobel2016multidimensional} shows that this algorithm has regret at least $\Omega(d^2)$.
Inflating the set by taking the Minkowski sum with a ball seems to be the appropriate regularization that allows us to overcome the $d^2$ lower bound.

Another natural algorithm is to guess $y_t = \frac{1}{2}\left(\min_{v \in S_t} \dot{v}{x_t} +  \max_{v \in S_t}\dot{v}{x_t} \right)$. This algorithm was shown to have regret at least $2^{\Omega(d)}$ in \cite{CohenLL16}.

\subsection{Pricing loss with $O(d \log \log T +  d\log d)$ regret}\label{subsec:dloglogt}
We now study the pricing loss $\ell(y_t, u_t) = u_t - y_t \cdot \one\{y_t \leq u_t \}$. Unlike the previous case the loss function is discontinuous. While the loss when under-estimating $u_t$ is small, the loss when over-estimating is very large. In the one-dimensional setting, \cite{kleinberg2003value} obtains a $O(\log \log T)$-regret algorithm by a conservative variant of binary search that avoids over-estimating the actual value as much as possible.

As before, our algorithm will keep track of $\Vol(S_t + z_i \B)$ for different values $z_i$. This time, however, we will guess more conservatively so that in the case of a no-purchase event, the potential will decrease by a large amount. We will do this in a way that each $z_i$ can lead to a no-purchase event approximately $O(d)$ times.

\begin{algorithm}[H]
\caption{{\sc Multiscale Steiner Potential for Pricing} }
\label{alg:linear_pricing}
\begin{algorithmic} 
\State Initialize $S_1 = \B$ and let $z_i = 2^{-3 \cdot 2^{i}} / (16d)$  for all $i$ 
\For {$t$ in $1,2, \dots , T$}
\State Adversary picks $x_t$
\If {$\width(S_t; x_t) \leq 1/T$} 
\State Query $y_t = \min_{v \in S_t} \dot{v}{x_t}$.
\Else
\State Let $i$ be the largest index such that $\width(S_t; x_t) \leq 2^{-2^i}$
\State Let $m_t$ such that $\Vol(\{v \in S_t + z_i \B; \dot{v}{x_t} \leq m_t\}) = 2^{-2^{i-1}} \Vol(S_t + z_i \B)$
\State Query $y_t = m_t - z_i$
\EndIf
\EndFor
\end{algorithmic}
\end{algorithm}

\begin{lemma}\label{lem:pricing_loss_over}
If $\width(S_t; x_t) > 1/T$ and $y_t > \dot{v_0}{x_t}$ (resulting in a no-purchase event) then $$\Vol(S_{t+1} + z_i \B) \leq 2^{-2^{i-1}} \Vol(S_t + z_i \B)$$
\end{lemma}
\begin{proof}
Note that in this case, the set $\{v \in S_t + z_i \B; \dot{v}{x_t} > m_t\}$ is disjoint from $S_{t+1} + z_iB$.  The definition of $m_t$ immediately gives the desired result.
\end{proof}
\begin{lemma}\label{lem:pricing_loss_under}
If $\width(S_t; x_t) > 1/T$ and $y_t \leq \dot{v_0}{x_t}$ (resulting in a purchase event) then $$\Vol(S_{t+1} + z_i \B) \leq \left(1 - \frac{1}{10 \cdot 2^{2^{(i-1)}}}\right) \Vol(S_t + z_i \B)$$
\end{lemma}
\begin{proof}
First we upper bound the volume $V$ of the strip $ \Vol(\{v \in S_t + z_i \B; m_t - 2z_i \leq \dot{v}{x_t} \leq m_t\})$.  Let $C$ be the largest volume of a section of $S_t + z_i \B$ in the direction $x_t$ (see right part of Figure \ref{fig:dlogd}).  Then $V \leq 2 C z_i$.  On the other hand since $S_t + z_iB$ is a convex set and has width at least $2^{-2^{i+1}}$, 
$$\Vol(S_t + z_i\B) \geq \frac{C \cdot 2^{-2^{i+1}}}{d}.$$
Thus
\begin{align*}
\Vol(\{v \in S_t + z_i \B; \dot{v}{x_t} \leq m_t -2 z_i \}) & \geq 2^{-2^{(i-1)}} \cdot \Vol(S_t + z_i\B)  - \frac{2z_id}{2^{-2^{i+1}}}\Vol(S_t + z_i\B)
\\ & \geq \frac{1}{10\cdot 2^{2^{i-1}}}\Vol(S_t + z_i\B)
\end{align*}
Since the entire set $\{v \in S_t + z_i \B; \dot{v}{x_t} \leq m_t -2 z_i \}$ is disjoint from $S_{t+1} + z_i \B$, this completes the proof of the lemma.
\end{proof}

\begin{theorem}\label{thm:pricingloss}
The  regret of the {\sc Multiscale Steiner Potential for Pricing} algorithm is at most $O(d \log \log T + d \log d)$.
\end{theorem}
\begin{proof}
First, note that the regret contributed from all rounds where $\width(S_t; x_t) \leq 1/T$ is $O(1)$.  Next, we consider rounds with $\width(S_t; x_t) > 1/T$.  Note that for all of these rounds $i \leq 2 \log \log T$.  If we choose index $i$ and  $y_t > \dot{v_0}{x_t}$ (leading to a no-purchase), the volume of $S_t + z_iB$ is cut by a factor of $2^{-2^{i-1}}$.  Since $\Vol(S_t + z_iB) \geq z_i^d \Vol(B)$ always, this can happen at most $$n_i^- := O\left( \frac{d\log (1/z_i)}{\log 2^{2^{i-1}}}\right) = O\left( \frac{d \left(  
\log d + 3 \cdot 2^i\right)}{2^{i-1}}\right) = O\left( \frac{d \log d} {2^{i-1}} + d \right) $$
times.  The loss from each such query is at most $1$.  If we choose index $i$ and $y_t \leq \dot{v_0}{x_t}$, the volume of $S_t + z_iB$ is cut by a factor of $\left( 1 - \frac{1}{10 \cdot 2^{2^{i-1}}}\right)$.  This can happen at most $$n_i^+ :=  O\left(\frac{d\log (1/z_i)}{-\log \left( 1 - \frac{1}{10 \cdot 2^{2^{i-1}}}\right)} \right) \leq \left(10 \cdot 2^{2^{i-1}} \cdot d\log (1/z_i) \right) = O \left( d \log d \cdot 2^{2^{i-1}} + 2^{i + 2^{i-1}} \right)$$ times.  The loss from each such query is at most $2^{-2^i}$.  Thus, our regret is at most
\begin{align*}
O(1) + \sum_{i=1}^{2\log \log T} (n_i^- + 2^{-2^i} \cdot n_i^+) = O(d \log \log T + d \log d )
\end{align*}
\end{proof}

\subsection{Polynomial time implementation}\label{sec:efficiency}
We note that $y_t$ can be approximated using binary search as long as we can compute the volume of $(S_t  + z_i \B)\cap H$ for any half-space $H$. It is enough to notice that we only need a constant approximation of the volume in the previous proof and in order to approximate the volume we only need access to a separation oracle \cite{KLS97,lovasz1999hit,BertsimasV04}. Since $S_t$ is a ball intersected with at most $t$ halfspaces, it is trivial to obtain a separation oracle for it.

To obtain a separation oracle for $S_t + z_i \B$ it is enough to solve the problem of computing the distance from a query point to the convex set $S_t$ which is itself a convex problem.
Technically, this requires $S_t$ to not be too small since the guarantee of cutting plane methods (like the ellipsoid algorithm) tells us that (given an initial ellipsoid $E \supseteq S_t$), it is possible to compute an $\epsilon$-optimal solution in time $O\left( \mathcal{T} \cdot \poly(d) \cdot \log\left( \frac{\Vol (E)}{\epsilon \Vol (S_t)} \right) \right)$.

We can take $E$ to be the unit ball; then this is algorithm is efficient as long as $\Vol(S_t)$ is never too small (anything at least $\exp(-\poly(T))$ is fine). Here we present a simple modification of Algorithms \ref{alg:linear_symmetric} and  \ref{alg:linear_pricing} that makes sure that the volume of $S_t$ stays large enough throughout by preserving a small ball around $v_0$.

Initialize $S_1 = (1+T^{-4}) \B$ to ensure that $\B(v_0, T^{-4}) \subseteq S_1$ (where $\B(c, r)$ denotes the ball with center $c$ and radius $r$). Now change the guess $y_t$ to $y_t - \delta_t$ where $\delta_t$ is sampled from the uniform distribution over $[0,T^{-2}]$. The total additional loss from this perturbation is $O(1)$. Since the perturbation is smaller than $z_i$ we can use the same argument in Lemmas \ref{lemma:volume_progress_symmetric} and \ref{lem:pricing_loss_under} with $2z_i$ instead of $z_i$ to bound the volume of the band making sure there is constant progress in the Steiner potential.

The advantage of this perturbation is that the probability that the cut passes through the ball of radius $1/T^4$ around $v_0$ is at most $1/T^2$ per period. So with probability $1-1/T$, $B(v_0, 1/T^4) \subseteq S_t$ for all periods $t$. It follows that $\Vol(S_t) \geq \frac{1}{T^{4d}}\Vol(\B)$, and with this, the convex minimization problem can be solved in $\poly(d,T)$ time.

\subsection{Optimality}\label{sec:optimality}

Here we briefly discuss the optimality of our results, namely Theorem \ref{thm:symloss} and Theorem \ref{thm:pricingloss}.  Note that we set up the problem by assuming that the hidden vector $v_0$ and all of the adversarially chosen contexts $x_t$ are drawn from the $L^2$ unit ball.  We may alternatively set up the problem by assuming that the hidden vector $v$ is drawn from the cube $[-1,1]^d$ and the contexts $x_t$ are chosen from the $L^1$ ball i.e. $\norm{x_t} \leq 1$ for all $t$.  With this setup, it follows from a direct reduction to the one-dimensional case that we cannot do better than $O(d)$ for symmetric loss and $O(d \log \log T)$ for pricing loss.  

Now we show how our proofs can be modified to achieve the same results, $O(d \log d)$ for symmetric loss and $O(d \log \log T + d \log d)$ for pricing loss, which are optimal in this modified setting up to logarithmic factors.  We will run exactly the same algorithms.  To see that the analysis remains the same, it suffices to note that our maximum loss in any round is still bounded by $O(1)$ and the volumes are scaled by a factor of $\poly(d)^d$ which becomes an additive $d \log d$ factor after taking the logarithm.

\section{Framework for Online Learning with Binary Feedback}\label{sec_setup}

We will now consider a general model of learning with binary feedback that has contextual search as a special case. We will derive regret bounds based on the covering dimension of the hypothesis class. The driving technique will still be the Steiner potential together with two new ideas: (i) a new adaptive scaling that will either make a lot of progress in a finer granularity or slower progress in coarser granularity; (ii) randomized cuts that will reduce the potential with constant probability.\\

Consider a hypothesis class $\H$ consisting of functions mapping $\X$ to $\Y$. We refer to $\X$ as the context space and $\Y$ as the output space.
We assume that the output space $\Y$ is a totally ordered set, i.e,  for each $y_1, y_2 \in \Y$ with $y_1 \neq y_2$ we have either $y_1 < y_2$ or $y_2 < y_1$.

The learning protocol is as follows: an adversary chooses some $f_0 \in H$ and in each round they choose some $x_t \in \X$.  The learner makes a prediction $y_t \in \Y$ and incurs loss $\ell(y_t, f_0(x_t))$ for some loss function $\ell: \Y \times \Y \rightarrow [0,1]$.  Upon making a prediction, the learner receives feedback on whether $y_t \leq f_0(x_t)$ or $y_t \geq f_0(x_t)$ (the feedback is arbitrary in case of equality). It will be convenient to represent the feedback as a variable $\sigma_t \in \{-1,+1\}$ such that $\sigma_t = +1$ if $y_t > f_{0}(x_t)$ and $\sigma_t = -1$ otherwise. 

We make the following assumptions about the loss function throughout the paper:
\begin{itemize}
    \item \emph{Reflexive:} $\ell(y,y) = 0$ for all $y \in \Y$.
    \item \emph{Symmetry:} $\ell(y_1,y_2) = \ell(y_2,y_1)$ for all $y_1,y_2 \in \Y$.
    \item \emph{Triangle inequality:} $\ell(y_1, y_2) \leq \ell(y_1, y') + \ell(y', y_2)$ for all $y_1, y_2, y' \in \Y$.
    \item \emph{Order consistency:} If $y_1 < y_2 < y_3$ then $\max \{ \ell(y_1, y_2), \ell(y_2, y_3) \} \leq \ell(y_1, y_3)$.
    \item \emph{Continuity:} If $0 < \ell < \ell(y_1, y_2)$ then there are $y', y'' \in \Y$ such that $\ell = \ell(y_1, y') = \ell(y'', y_2)$.
\end{itemize}

If $\Y = \R$, then for any continuous increasing function $\phi:\R \rightarrow [0,1]$ and parameter $\alpha \leq 1$, the loss $\ell(y_1, y_2) = \abs{\phi(y_1) - \phi(y_2)}^\alpha$ satisfies the desired properties. Note that while the symmetric loss function is easily cast in this framework the pricing loss is not captured. In Section \ref{subsec:impossibility_pricing} we give an impossibility result showing it is impossible to obtain covering-dimension bounds for the pricing loss.

\subsection{Covering Dimension}
Our loss bounds will be in terms of the covering dimension of the hypothesis class $\H$. We start by defining a metric $d_{\infty}(\cdot, \cdot)$ on $\H$ induced by the loss function: 

\begin{definition}
For two hypotheses $f_1,f_2 \in  \H$, let $d_{\infty}(f_1,f_2) = \sup_{x \in \mathcal{X}}\left(\ell(f_1(x), f_2(x)) \right)$
\end{definition}

 We can now introduce the notions of $\eps$-net and covering dimension.

\begin{definition}[$\eps$-net]
For an $\eps$, we say that a set $S \subseteq \H$ is an $\epsilon$-net of $\H$ under the  $d_{\infty}$ metric if for every $h \in \H$, there is $h' \in S$ with $d_{\infty}(h,h') \leq \epsilon$. Let $N_{\eps}(\H)$ be an $\epsilon$-net of $\H$ of minimum cardinality.
\end{definition}

\begin{definition}[Covering Dimension]\label{def:covering_dim}
Define the covering dimension $\H$ as
\[
\Cdim(\H) = \sup_{0 < \epsilon \leq \frac{1}{2}} \frac{\log \abs{N_{\eps}(\H)}}{\log \frac{1}{\eps}}
\]

Note that this definition of $\Cdim$ differs from Hausdorff dimension in that we care not just about the limit $\eps \rightarrow 0$, but the largest value for any $\eps \in [0, 1/2]$); importantly, this guarantees us that, for any $\eps \in (0, 1)$,

\[
N_{\eps}(\H) \leq \max(\eps^{-1}, 2)^{\Cdim(\H)}.
\]
Note that we specify $\epsilon \leq \frac{1}{2}$ to avoid issues when $\epsilon$ is close to $1$ - any other fixed constant upper bound $p$ only changes this dimension by a constant factor of at most $1 / \log (1/p)$.
\end{definition}

We give a few quick examples to give intuition about covering dimension.
\begin{example}
The space of functions $f:[d] \rightarrow \{0,1\}$ with loss function $\ell(y_1,y_2) = 1_{y_1 \neq y_2}$ has covering dimension $d$.
\end{example}
\begin{example}[Contextual Search]\label{contextual_search}
Let $\B$ be the unit-ball in $\R^d$, i.e. $\B = \{x \in \R^d; \norm{x}_2 \leq 1 \}$. For each $v \in \B$, let $f_v : \B \rightarrow \R$ be defined by the dot product $f_v(x) = \dot{v}{x}$.  The linear contextual search problem is defined as the learning problem for class $\H = \{f_v; v \in \B\}$ with loss function $\ell(y_1, y_2) = \abs{y_1 - y_2}$.  This class has covering dimension $O(d)$.

To see that the covering dimension is $O(d)$, first note that $d_\infty(f_v, f_u) = \norm{u-v}_2$.  It suffices to show that for any $0 < \epsilon \leq \frac{1}{2}$, there is an $\epsilon$-net of the sphere of size $\left(1/\epsilon\right)^{O(d)}$.  To do this, we can greedily place points in the unit ball such that no two are within $\epsilon$ of each other.  If we draw an $\frac{\epsilon}{2}$-radius ball around each point, these balls must be disjoint and contained in a ball centered at the origin of radius $1 + \frac{\epsilon}{2}$.  Thus, the maximum number of points we will place is 
\[
\left(\frac{1 + \frac{\epsilon}{2}}{\frac{\epsilon}{2}} \right)^d = \left(1 + \frac{2}{\epsilon} \right)^d
\]
giving us an $\epsilon$-net of the same size.
\end{example}

\begin{example}[Sparse Contextual Search]\label{ex:sparse} The sparse version of the contextual search problem is given by class $\H = \{f_v; v \in \B, \norm{v}_0 \leq s\}$ where $\norm{v}_0 := \abs{\left\{ i; v_i \neq 0 \right\}}$.  The loss function is still $\ell(y_1,y_2) = |y_1 - y_2|$. The covering dimension of this class is $O(s \log d)$ (where we treat $s$ as a constant and $d$ as tending to $\infty$).
To see that the covering dimension is $O(s \log d)$, note that for any $\epsilon$ we can use the result in the previous example to obtain an $\epsilon$-net of size
\[
\binom{d}{s}\left(1 + \frac{2}{\epsilon} \right)^s \leq s^d\left(1 + \frac{2}{\epsilon} \right)^s.
\]

\end{example}

\begin{example}[Unit demand] In the unit demand version of contextual search the set $\X = \{0,1\}^d$ and the hypothesis class consists of functions $f_w(x) = \max_{i \in [d]} w_i x_i$ for $w \in [0,1]^d$. The covering dimension of this class is $O(d)$. This example corresponds to the puzzle in the introduction and corresponds to the economic situation where a seller wants to price a bundle of goods (represented by the context) but the buyer has an unit-demand valuation, i.e., only cares about the highest-valued item in the bundle. 

To see that the covering dimension is $O(d)$, note that the set $S = \{\epsilon x| x \in \{0,1, \dots , \lfloor 1/\epsilon \rfloor\}^d \}$ forms an $\epsilon$-net and $|S| \sim \left( \frac{1}{\epsilon}\right)^d$.
\end{example}

\section{Loss Bounds based on Covering Dimension}

For simplicity, in the following theorems we will assume our hypotheses map $\mathcal{X} \rightarrow \R$ and that our loss is given by $\ell(y,y') = \abs{y-y'}$.  We will then remark on how to generalize our proof to other loss functions.

\subsection{$O(d \log(T))$ regret via Single-scale Steiner Potential}\label{subsec:dlogt_noiseless}

Bounds that depend on $\log(T)$ are typically easy for learning with binary feedback and can be obtained using different algorithms. The interesting question in this setting is how to obtain bounds that are constant in $T$. Nevertheless, it is  instructive to start with a simpler algorithm with regret $O(d \log T)$ for $d = \Cdim(\H)$. It will illustrate how the \emph{Steiner potential} can be generalized to an abstract setting. Instead of keeping track of the hypotheses that are consistent with the feedback so far, we will keep an inflated version of that set.

This algorithm starts with a $T^{-1}$-net of the hypothesis class and keeps a set of candidate hypotheses that are approximately consistent with observations seen so far. For each $x_t$, it queries a random point around the median, halving the set of hypotheses with at least half probability.

\begin{algorithm}[H]
\caption{{\sc Single-scale Steiner Potential} }
\begin{algorithmic}

\State Initialize $H_1 = N_{1/T}(\H)$ 
\For {$t$ in $1,2, \dots , T$}
\State Adversary picks $x_t$
\State Let $m_t$ be the median of $\{f(x_t); f \in H_t\}$
\State Choose $y_t = m_t - 2/T$ or $m_t + 2/T$, each with probability half
\State Update $H_{t+1} = \{f \in H_t; \sigma_t(y_t - f(x_t)) \geq -1/T \}$
\EndFor

\end{algorithmic}
\end{algorithm}

The analysis is based on the following lemma:

\begin{lemma}\label{lemma:half_hypothesis}
If $\abs{m_t - f_0(x_t)} > 2/T$, then $\abs{H_t} \leq \frac{1}{2} \abs{H_{t-1}}$ with probability at least $1/2$.
\end{lemma}

\begin{proof}
Assume $f_0(x_t) > m_t + 2/T$ (the other case is analogous), then with probability half, the algorithm guesses $y_t = m_t + 2/T$ and gets the feedback that $f_0(x_t) \geq y_t$, which eliminates all the hypotheses $f$ such that $f(x_t) < m_t + 1/T$. These hypotheses constitute at least half of $H_t$. 
\end{proof}

\begin{corollary}
The \textsc{Single-scale Steiner Potential} algorithm obtains regret $O(d \log T)$ in expectation.
\end{corollary}

\begin{proof}
The regret from periods where $\abs{m_t - f_0(x_t)} \leq 2/T$ is at most $O(1)$. For the remaining periods, the size of $\abs{H_t}$ is halved with at least $\frac{1}{2}$ probability. Note that for such a $t$,
\[
\E\left[ \frac{1}{|H_{t+1}|} \right] \geq \frac{3}{2}\E\left[\frac{1}{|H_t|} \right]
\]
Next, $\abs{H_t} \geq 1$ for all $t$ since there is some element in $H_1$ that is $\frac{1}{T}$-close to $f_0$ which is never eliminated.  However, $\frac{1}{|H_1|} \geq \frac{1}{T^d}$ and $\frac{1}{|H_t|} \leq 1$ for all $t$.  Thus for any integer $c$, the probability that there are $c$ periods with loss greater than $\frac{2}{T}$ is at most $\frac{T^d}{\left(\frac{3}{2}\right)^c}$.  Thus, the expected number of periods with loss larger than $2/T$ is at most
\[
10d\log T + \sum_{i=\lceil 10d\log T \rceil}^{\infty}\frac{T^d}{\left(\frac{3}{2}\right)^i} = O(d \log T).
\]
\end{proof}

\subsection{Strategy for achieving constant (in $T$) regret}\label{sec:intuition_const_T}

In this subsection we provide some intuition on how to improve the regret of our algorithm from $O(d \log T)$ to $O(d^2)$. In  \textsc{Single-scale Steiner Potential} a loss larger than $1/T$ causes $H_t$ to half in size, but whenever it halves in size the only bound we can get for the loss is $1$. To improve this bound, we need to guarantee that a loss of $1$ can't occur very often. Our strategy for doing that involves keeping multiple levels of discretization. Given $y_t$ and the feedback $\sigma_t \in \{-1,+1\}$ we will keep for each $z > 0$:
$$H_1^z = N_z(\H) \quad \text{and} \quad H_{t+1}^z = \{ f \in H_{t}^z; \sigma_t (y_t - f(x_t)) \geq -z \}$$

In other words, we keep an $z$-discretization of hypotheses along with all the hypotheses that are consistent with the feedback so far with an $z$-margin. The $z$-margin is important to guarantee that any hypothesis that is $z$-close to $f_0$ will never be eliminated. We will also refer to $H_t^0$ as the set of hypotheses consistent with observations so far without any discretization or margin.

Our strategy will be to choose in each round some discretization level $z$ based on the maximum possible loss achievable in this round. We will divide the space of losses in exponentially sized buckets and define $i$ to be the index of the bucket where the maximum loss falls: 
$$\width(H_t^0; x_t) := \max_{f \in H_t^0} f(x_t) - \min_{f \in H_t^0} f(x_t) \in \left( 10 \cdot 2^{-(i+1)}, 10 \cdot  2^{-i} \right].$$

Now we choose $z = z_i$ based on $i$, compute the median $m_t$ of $\{f(x_t); f \in H^{z_i}_t\}$ and guess either $y_t = m_t - 2z_i$ or $m_t + 2z_i$ with half probability each. Now one of two things can happen:

\begin{itemize}
    \item If the loss is larger than $2 z_i$, the set $H_t^{z_i}$ will decrease by a factor of $2$ with half probability. This should happen at most $\log \abs{N_{z_i}}$ times in expectation, generating loss $10 \cdot 2^{-i} \log \abs{N_{z_i}} = O( 2^{-i} \cdot d \log(1/z_i) )$.
    \item If the loss is smaller than $2 z_i$ we will show that the set $H_t^{2^{-i}}$ will decrease by at least $1$ element in expectation (Lemma \ref{lemma:small_loss}), so we get loss $z_i \cdot \abs{N_{2^{-i}}} = O( z_i 2^{d i} )$ 
\end{itemize}

This leads to a regret of:
$$O \left( \sum_i 2^{-i} \cdot d \log(1/z_i) + z_i 2^{d i} \right) = O(d^2) \quad \text{for} \quad z_i = \frac{1}{3^{d(i+1)}}$$

\subsection{Analysis of the $O(d^2)$ algorithm}\label{subsec:d_square_noiseless}

\begin{theorem}\label{main_binaryfeedback}
Let $d = \Cdim(H)$.  The \textsc{Multi-scale Steiner Potential} algorithm incurs expected regret $O(d^2)$ in the binary feedback model.
\end{theorem}

\begin{algorithm}[H]
\caption{{\sc General Multi-scale Steiner Potential} }
\begin{algorithmic}
\State Let $z_i = {3^{-d(i+1)}}$ for all $i$ 
\For {$t$ in $1,2, \dots , T$}
\State Adversary picks $x_t$
\State Let $i$ be the largest index such that $\width(H^0_t ; x_t) \leq 10 \cdot 2^{-i}$
\State Let $m_t$ be the median of the set $\{ f(x_t) | f \in H^{z_i}_t \}$
\State Query either $m_t + 2z_{i}$ or $m_t - 2z_{i}$ each with half probability
\EndFor

\end{algorithmic}
\end{algorithm}

Note if there does not exist an index $i$ such that $\width(H_t^0; x_t)  \leq 10 \cdot 2^{-i}$ then we must actually have $\max_{f \in H_t^0}f(x_t) = \min_{f \in H_t^0}f(x_t) = f_0(x_t)$.  In this case we know the value of $f_0(x_t)$ for sure so we simply query this value and incur $0$ loss.

\begin{lemma}\label{lemma:big_loss}
If $i$ is the index chosen in the $t$-th step and $\abs{m_t - f_0(x_t)} > 2z_i$ then $\abs{H^{z_i}_t} \leq \frac{1}{2} \abs{H^{z_i}_{t+1}}$ with probability at least $1/2$.
\end{lemma}
\begin{proof}
The same as the proof of Lemma \ref{lemma:half_hypothesis} replacing $1/T$ by $z_i$.
\end{proof}

The new ingredient is a ``potential" argument when the loss is small:

\begin{lemma}\label{lemma:small_loss}
If $i$ is the index chosen in the $t$-th step and $\abs{m_t - f_0(x_t)} \leq 2z_i$, then with probability at least $1/2$, $\abs{H^{r}_t} \leq \abs{H^{r}_{t+1}} -1$ for $r = 2^{-(i+1)}$
\end{lemma}

\begin{proof} By the choice of the index $i$, $\width(H^0_t; x_t) > 10\cdot 2^{-(i+1)} = 10 r$, so there must exist $f \in H_0^t$ such that $\abs{f(x_t) - m_t} \geq 5r$. Let's assume that $f(x_t) \geq m_t + 5 r$ (the other case is analogous). The algorithm will query $m_t + 2z_i$ with half probability and learn that $f_0(x_t) \leq m_t + 2z_i$ (by the assumption that $\abs{m_t - f_0(x_t)} \leq 2z_i$).

Such a query must eliminate some hypothesis $f'  \in H_t^{r}$ since there must be some $f' \in H_t^{r}$ with $d_{\infty}(f,f') \leq r$, so this hypothesis must satisfy $f'(x_t) \geq m_t + 4 \cdot r$ and hence will be ruled out by the information from querying $m_t + 2z_{i}$.
\end{proof}

We can now proceed to prove Theorem \ref{main_binaryfeedback}.

\begin{proof}[Proof of Theorem \ref{main_binaryfeedback}]
Let $A_i$ be the number of times that index $i$ is chosen by the algorithm and $|m_t - f_0(x_t)| > 2z_i$.  Let $B_i$ be the number of times that index $i$ is chosen and $|m_t - f_0(x_t)| \leq 2z_i$.  Combining the previous two claims (Lemmas \ref{lemma:big_loss} and \ref{lemma:small_loss}), we have that
\begin{align*}
\E[A_i] \leq 2\log_2 N_{z_i}(H) \\
\E[B_i] \leq 2N_{r_{i+1}}(H)
\end{align*}
where $r_i = 2^{-i}$.
For each query with index $i$, the loss is at most $10 r_i$.  Also for queries with $|m_t - f_0(x_t)| \leq 2z_i$, the loss is at most $2z_i$.  Thus the total loss is at most $\sum_i (10r_iA_i + 2z_iB_i)$.  It remains to note that
\begin{align*}
\E\left[\sum_i (10r_iA_i + 2z_iB_i)\right] \leq \left( \sum_{i=1}^{\infty} 20r_i\log_2 N_{z_i}(H) + 4z_iN_{v_{i+1}}(H) \right) = O(d^2).
\end{align*}
\end{proof}

\begin{corollary}
In the Contextual Search with Symmetric Loss, if the hidden vector $v \in \R^d$ is guaranteed to be $s$-sparse then there is an algorithm with total regret $O(s^2 \log^2(d))$.
\end{corollary}

\begin{proof}
By combining Theorem \ref{main_binaryfeedback} and Example \ref{ex:sparse}.
\end{proof}

\begin{remark}
To adjust our proof to deal with any loss functions satisfying the assumptions outlined in Section \ref{sec_setup} we make the following adjustments.  We replace $\max_{f \in H_t^0}f(x_t) - \min_{f \in H_t^0}f(x_t)$ with $L(\max_{f \in H_t^0}f(x_t),\min_{f \in H_t^0}f(x_t))$.  Also, we replace $m_t + 2z_{i_t}$ with any $y \in \mathcal{Y}$ such that $m_t < y$ and $L(m_t,y) = 2z_{i_t}$ and similar for $m_t - 2z_{i_t}$ (note this $y$ exists by the continuity of the loss function).  
\end{remark}

\subsection{Impossibility Results for Pricing Loss}\label{subsec:impossibility_pricing}
The results from the previous section apply for loss functions that are somewhat well-behaved i.e. satisfying the conditions outlined at the beginning of Section \ref{sec_setup}.  Clearly, the pricing loss function does not satisfy these assumptions.  While one may hope to guarantee $\poly(d)\log \log T$ total loss (where $d$ is the covering dimension of the hypothesis class), here we show that for the pricing loss function, it is actually impossible to guarantee regret in this case that is polynomial in the covering dimension.

\begin{claim}
Let $\B$ be the unit ball in $\R^d$.  Consider the domain $\mathcal{X} = \B$ and the hypothesis class $\mathcal{H} = \{f_v |  v \in \R^d,  ||v||_0 = 1, ||v||_2 \leq 1 \}$ where $f_v(x) = \dot{v}{x}$  (note this is the same as the hypothesis class in Example \ref{ex:sparse} with $s = 1$).  Any learner must incur at least $\Omega\left(\sqrt{d}\right)$ regret over $d^2$ rounds with the pricing loss function. 
\end{claim}
\begin{remark}
Note that the covering dimension of $\mathcal{H}$ is $O(\log d)$ so the above claim implies that there is an exponential separation between the regret (in the pricing loss setting) and covering dimension.
\end{remark}
\begin{proof}
Choose $v$ uniformly at random from the $d$ points $(1,0, \dots , 0), (0,1, \dots , 0), \dots (0, \dots , 0, 1)$ and let the true function be $f_v$.  Now let $$x_1 = \left( \frac{1}{\sqrt{2}}, \frac{1}{\sqrt{2(d-1)}}, \frac{1}{\sqrt{2(d-1)}}, \dots , \frac{1}{\sqrt{2(d-1)}} \right)$$
Let $x_2, \dots , x_d$ be obtained by cyclically permuting the coordinates of $x_1$.  Now the adversary randomly permutes $x_1, \dots , x_d$ to obtain a sequence $x_{i_1}, \dots , x_{i_d}$ and presents the points in that order to the learner.  Note the learner gains no information when it guesses a value $y \leq \frac{1}{\sqrt{2(d-1)}}$.  If all of the learner's guesses through the first $d$ rounds are at most $\frac{1}{\sqrt{2(d-1)}}$ then the learner incurs loss at least $\frac{1}{3}$ over the first $d$ rounds and has gained no information.  The adversary can then repeat this process.
\\\\
Now it remains to consider when the learner guesses a value above $\frac{1}{\sqrt{2(d-1)}}$ at some point within the first $d$ rounds.  Say the first time this occurs is at round $j$.  With $\frac{d-1}{d}$ probability, the learner incurs $\frac{1}{\sqrt{2(d-1)}}$ loss this round and is able to eliminate one of the $d$ possible hypotheses.  The problem then effectively reduces to $d-1$ dimensions.  Repeating this argument inductively, we see that the adversary can guarantee regret
\[
\frac{d-1}{d}\left(\frac{1}{\sqrt{2(d-1)}} + \frac{d-2}{d-1} \left(\frac{1}{\sqrt{2(d-2)}} +  \dots     \right) \right) = \frac{1}{\sqrt{2}} \left( \frac{\sqrt{d-1}}{d} + \frac{\sqrt{d-2}}{d} + \dots  \right) = \Omega\left(\sqrt{d}\right).
\]

\end{proof}

\section{Noisy Feedback}

We now consider the binary feedback model with noise, where each round the feedback is (independently) flipped with probability $p < 1/2$. We will no longer be able to eliminate a hypothesis based on the feedback since it is always possible that the feedback was flipped, instead we will keep a weight function expressing the likelihood of each hypothesis given observations.

We start by giving a $O(d \log T)$ algorithm for general hypothesis classes. The algorithmic techniques will be standard and inspired by both Bayesian Inference and by algorithms in the Multiplicative Weight Updates family (such a Hedge or Weighted Majority). The analysis, however, will deviate from the usual analysis of multiplicative weights given the type of feedback. We don't have access to the actual loss of the arm we pulled nor an unbiased estimator thereof. This will require both a new potential function as well as a modification of the multiplicative weights framework: instead of sampling from the distribution induced by the weights, we will get the (weighted) median advice on what is the right guess for this context.

Our main innovation is in Section \ref{subsec:polyd_noisy} where we obtain an algorithm with $O(\poly(d))$ regret (independent of $T$) for the linear contextual search case. Instead of one weight function, we keep a family of weight functions. Each weight function will correspond to different levels of uncertainty about the inner product of the hidden point with the current context. By using geometric techniques to analyze the stochastic evolution of the weights, we show that they must concentrate near the true hypothesis. 

\subsection{ $O(d \log T)$ algorithm for a general hypothesis class }\label{sec:dlogt_noisy}

 The usual approach in Bayesian inference is to start with a uniform prior over the set of hypotheses and given each observation, compute the posterior. It is important to emphasize that the true hypothesis $f_0$ in our model is still chosen adversarially. The Bayesian inference only serves to provide the intuition.

The algorithm will be as follows: as before we will start with a discretized version of the hypothesis class  $N_{T^{-2}}(H)$ which we will call $N$ for short in this section. We will keep a weight function $w_t: N \rightarrow \R$ which roughly expresses the likelihood that a hypothesis is close to the true hypothesis. Our guess will be a perturbed version of the weighted median $m_t$ of the set $\{f(x_t); f \in N\}$. Formally, the weighted median $m_t$ is a number that satisfies:
$$\sum_{f \in N; f(x_t) \geq m_t} w_t(f) \geq \frac{1}{2} \quad \text{and} \quad
\sum_{f \in N; f(x_t) \leq m_t} w_t(f) \geq \frac{1}{2} $$
After receiving the feedback, we will update the weights in the following way (we will choose $y_t$ at random such that $y_t = f(x_t)$ occurs with zero probability):
$$\hat{w}_{t+1}(f) = \left\{ \begin{aligned} 
    & p' \cdot w_t(f) \quad  & & \text{if } \sigma_t (y_t - f(x_t)) < 0 \\
    & (1-p') \cdot w_t(f) \quad  & & \text{if } \sigma_t (y_t - f(x_t)) > 0 
  \end{aligned} \right. $$
for some parameter $p'$ and re-normalizing afterwards:
$$w_{t+1}(f) = \frac{\hat{w}_{t+1}(f)}{\sum_{f' \in N} \hat w_{t+1}(f')}$$
In standard Bayesian inference, we would normally use $p = p'$. For this algorithm, we will choose any parameter $p'$ with $p < p' < 1/2$. The actual choice of parameter will only affect the constants. Also note that unlike Bayesian inference we don't choose the guess with largest likelihood but a perturbed version of the median.

\begin{algorithm}[H]
\caption{{\sc Single-Scale Steiner Potential with Noise}}
\label{alg:general_noise}
\begin{algorithmic}
\State Initialize $w_1(f) = 1/\abs{N}$ for all $f \in N := N_{T^{-2}}(H)$.
\For {$t$ in $1,2, \dots , T$}
\State Adversary picks $x_t$
\State Let $m_t$ be the weighted median $\{f(x_t); f \in N\}$ with respect to weights $w_t$.
\State Choose $y_t \in \left[ m_t - \frac{1}{T}, m_t + \frac{1}{T} \right]$ uniformly at random
\State Choose $\hat w_{t+1}(f) = p' \cdot w_t(f)$ if $\sigma_t(y_t - f(x_t)) < 0$ and  $\hat w_{t+1}(f) = (1-p') \cdot w_t(f)$ otherwise
\State Normalize the weights: $w_{t+1}(f) = \hat w_{t+1}(f) / [\sum_{f' \in N}  \hat w_{t+1}(f')] $

\EndFor
\end{algorithmic}
\end{algorithm}

\begin{theorem}\label{noisygeneral}
In the noisy feedback model, the above algorithm incurs expected regret $O(d \log T)$. 
\end{theorem}

We will denote the true hypothesis by $f_0$ as usual. Since $f_0$ might not belong to the discretized set $N$, we will control the weight that is placed on the closest hypothesis. Let $f_1$ be a hypothesis in $N$ with $d_{\infty}(f_1, f_0) \leq T^{-2}$.

\begin{lemma}\label{lemma:weight_increase_expectation}
If $f_1(x_t), f_0(x_t)$ are on the same side of $y_t$ and $W_t^-$ is the total weight mass on the other side of $y_t$, then we have the following equality:

$$\E\left[ \frac{1}{w_{t+1}(f_1)} \right] = \frac{1 - c W_t^-}{w_t(f_1)}.$$

Here the expectation is taken over the randomness in the feedback, and $c$ is a constant satisfying $0 < c < 1$ given by

$$c = (p'-p) \left[\frac{1-p'}{p'} - \frac{p'}{1-p'}\right].$$
\end{lemma}

\begin{proof}
Assume wlog that $f_1(x_t), f_0(x_t) \geq y_t$ and let $W_t^- = \sum_{\substack{f \in N; f(x_t) < y_t}} w_t(f)$ be the weight on hypotheses in the opposite direction and $W_t^+ = 1-W_t^-$ the remaining weight. Since $y_t$ is chosen from a continuous distribution, the event that some $f$ has $f(x_t) = y_t$ occurs with zero probability. With probability $1-p$ we have
$$w_{t+1}(f_1) = \frac{(1-p') \cdot  w_t(f_1)}{  (1-p') \cdot W_t^+ + p' \cdot W_t^- } $$
with the remaining probability $p$ we have:
$$w_{t+1}(f_1) = \frac{p' \cdot w_t(f_1)}{  p' \cdot  W_t^+ +  (1-p') \cdot W_t^- }$$
Averaging them we have:
$$\begin{aligned} \E\left[ \frac{1}{w_{t+1}(f_1)} \right] & = (1-p) \cdot \frac{  (1-p') \cdot W_t^+ + p' \cdot W_t^- }{(1-p') \cdot  w_t(f_1)} + p \cdot \frac{  p' \cdot  W_t^+ +  (1-p') \cdot W_t^- }{p' \cdot w_t(f_1)}  \\ 
& = \frac{1}{w_t(f_1)} \left[ W_t^+ + \left( \frac{p'}{1-p'}(1-p) + \frac{1-p'}{p'}p \right) W_t^- \right] = \frac{1-c W_t^-}{w_t(f_1)}
\end{aligned}$$
since $0 < \frac{p'}{1-p'}(1-p) + \frac{1-p'}{p'}p = 1 + (p'-p) \left[\frac{p'}{1-p'} - \frac{1-p'}{p'} \right] = 1 - c.$
\end{proof}

\begin{lemma} In expectation over both the randomness in the choice of $y_t$ and the randomness in the feedback, we have:
$$
\E\left[ \frac{1}{w_{t+1}(f_1)} \right] \leq \left(1 + \frac{c'}{T} \right) \cdot \frac{1}{w_t(f_1)} \quad \text{for} \quad c' = \frac{1-p'}{2p'} 
$$
\end{lemma}

\begin{proof}

With at least $1 - \frac{1}{2 T}$ probability, $f_1(x_t)$ and $f_0(x_t)$ are on the same side of $y_t$ since \[\abs{f_1(x_t) - f_0(x_t)} < 1/T^2\]
and the magnitude of the perturbation is $1/T$. In this case, by the previous lemma, we have: $\E \left[ {1}/{w_{t+1}(f_1)} \right] = {1}/{w_t(f_1)}$. With the remaining probability $f_1(x_t)$ and $f_0(x_t)$ are on opposite sides of $y_t$ and we can use the trivial bound in the equation below.  $$\frac{1}{w_{t+1}(f_1)} \leq \frac{1-p'}{p'} \frac{1}{w_t(f)}.$$ 
Combining, we get the desired inequality.
\end{proof}

\begin{lemma}

If $|f_0(x_t) - m_t| > \frac{2}{T}$ then for the constant $c$ in Lemma \ref{lemma:weight_increase_expectation}, then in expectation over both the
randomness in the choice of $y_t$ and the randomness in the feedback, we have:
$$ \E\left[ \frac{1}{w_{t+1}(f_1)} \right] \leq \frac{1-c/4}{w_t(f)}$$
\end{lemma}
\begin{proof}
Assume without loss of generality that $f_0(x_t) > m_t + \frac{2}{T}$. Since the magnitude of the perturbation is $1/T$, $f_1(x_t)$ will be on the same side of our guess as $f_0(x_t)$ and hence we can apply Lemma \ref{lemma:weight_increase_expectation}. With probability at least $1/2$ we have $y_t > m_t$ and hence $W_t^- \geq 1/2$. With the remaining probability we use the trivial bound $W_t^- \geq 0$. Combining those we get the bound in the statement.
\end{proof}

The previous lemmas imply that $w_t(f)$ grows by a constant factor (in expectation) whenever the median is far from the true point. We conclude the proof by showing that this can't happen too often since weights are bounded.

\begin{proof}[Proof of Theorem \ref{noisygeneral}]

The regret bound follows directly from the fact that  the probability of having $\abs{f_0(x_t) - m_t} > 2/T$ for more than $\Omega(d \log T)$ periods is at most $O(1/T)$.  Our strategy for proving this is to define a random process $Y_t$ that is a super-martingale, i.e. $\E[Y_{t+1}] \leq Y_t$ and argue that if $\abs{f_0(x_t) - m_t} > 2/T$ happens too often, then $Y_T$ will be much larger than $Y_1$.  This happens with small probability by Markov's inequality. 

We will first define an auxiliary sequence of real numbers $s_t \geq 0$ for $t \in \{1\hdots T\}$ as follows. Let $s_1 = 1/\abs{N} \geq T^{-2d}$ and let:
$$s_{t+1} = \left\{ \begin{aligned}
& (1-c/4)^{-1} \cdot s_t, \quad & & \text{if } \abs{f_0(x_t) - m_t} > 2/T \\
& (1+c'/T)^{-1} \cdot s_t, \quad & & \text{otherwise}
\end{aligned} \right.$$
Now define the following stochastic process:
$$Y_t = \frac{s_t}{w_t(f_1)}$$
It is simple to see that $Y_1 = 1$. The previous lemmas imply that $Y_t$ is a super-martingale, i.e. $\E[Y_{t+1}] \leq Y_t$ and hence $\E[Y_T] \leq 1$.  Now, in the case that  $\abs{f_0(x_t) - m_t} > 2/T$ for more than $\Omega(d \log T)$ periods, we have
$$s_T \geq T^{-2d} \cdot (1-c/4)^{-\Omega(d \log T)} \cdot (1+c'/T)^{-T} \geq \Omega(T)$$ 
and hence:
$$Y_T = \frac{s_T}{w_T(f_1)} \geq s_T \geq \Omega(T)$$
but since $\E[Y_T] \leq 1$, this can happen with at most $O(1/T)$ probability by Markov's inequality.
\end{proof}

\begin{remark}
We've assumed here that $p$ is a constant bounded away from $1/2$. How does the regret of this algorithm depend on $p$ as $p$ approaches $1/2$? If $p = \frac{1}{2} - \delta$, then we can set $p' = \frac{1}{2} - \delta'$ where $\delta' = \delta/2$. This leads to $c = O(\delta^{2})$ -- adapting the proof of Theorem \ref{noisygeneral} then shows we can have at most $O\left(\frac{d \log T}{\delta^2}\right)$ inaccurate rounds, for a total of at most $O\left(\frac{d \log T}{\delta^2}\right)$ regret. 
\end{remark}

\paragraph{Comparison with other approaches} It is worth comparing our algorithm with other learning techniques in the multiplicative weights family. The `experts' in our problem form a continuous set with a linear structure, which resembles the settings of Kalai and Vempala \cite{kalai2005efficient} and Abernathy et al \cite{abernethy2012interior}. In their setting, however, the optimal achievable regret is $O(\sqrt{T})$ while in our case we achieve $O(\log T)$. Another feature of our model is the stochasticity of the losses. With stochastic losses, Wei and Luo \cite{wei2018more} recently showed that multiplicative weight update algorithms achieve $O(\log T)$ regret when the learning rate is tuned properly, but their guarantees depend on the inverse of the gap between the two best arms. An important difference, however, is that in our setting we don't have access to the loss. We only learn whether our guess was too large or too small, which doesn't allow us to apply any of those algorithms.

\subsection{$O(\poly(d))$ algorithm for Noisy Contextual Search}\label{subsec:polyd_noisy}

In this section we study contextual search in the noisy feedback model. We show that here we can achieve total loss $O(\poly(d))$ independent of $T$ by exploiting the geometry of the Euclidean space.  Throughout this section we will use $q_0 \in \B$ to denote the true point, i.e. $f_0(x) = \dot{q_0}{x}$.\\

Our approach (Algorithm \ref{alg:noisy_contextual_search_alg}) builds off the Bayesian inference approach in the previous section (Algorithm \ref{alg:general_noise}) by combining it with the multi-scale discretization ideas in Section \ref{subsec:d_square_noiseless}. At a high (and slightly inaccurate) level, Algorithm \ref{alg:noisy_contextual_search_alg} works as follows. Throughout the algorithm, we maintain a distribution $w$ over the unit ball $\B(0, 1)$ where $w(q)$ represents the likelihood that $q$ is our true point $q_0$. Each round $t$, we are provided a direction $x_t$ by the adversary. We begin by measuring the ``width'' of our distribution in the direction $x_t$ -- i.e., the length of the smallest interval in this direction which contains almost all of the mass of our distribution $w$. Then (similarly as in Algorithm \ref{alg:general_noise}), we will guess a perturbed version of the median of $w$ in the direction $x_t$, where the size of the perturbation depends on the width. Finally, we multiplicatively update the distribution $w$, penalizing points on the wrong side of our guess (again, similarly as in Algorithm \ref{alg:general_noise}). \\


In the actual algorithm, we maintain a separate distribution $w_i$ for each possible scale $\gamma_i$ for the width (in particular, we are in scale $w_i$ if almost all of the mass of $w_{i-1}$ is concentrated in a small strip in direction $x_t$). This aids analysis in letting us guarantee we operate in each scale for at most a bounded number of rounds, which lets us bound the total loss of this algorithm.


\begin{algorithm}[H]
\caption{{\sc Noisy Linear Contextual Search} }
\label{alg:noisy_contextual_search_alg}
\begin{algorithmic}
\State \textbf{Initialization:}
\State Let $\eta = 1/(2d^{10})$.
\For {each integer $i > 0$}
\State Let $\beta_i = \frac{1}{2^{100di}}$ and $\gamma_i = \frac{1}{2^i}$. 
\State Construct a distribution $w_i: \B(0,1) \rightarrow \R$. Let $w_{i, t}$ denote the weight function $w_i$ at round $t$. \State Initialize $w_{i, 0}(q) = {1}/{\Vol(\B(0,1))}$ for all $q \in B(0, 1)$.
\State Initialize $C_i = 0$. ($C_i$ will store the number of times we have been in scale $i$). 
\EndFor
\State \textbf{Algorithm:}
\For {$t$ in $1,2, \dots , T$}
\State Adversary picks $x_t$.
\State Let $i_t$ be the largest index $i$ such that at least one of the following is true:
\begin{itemize}
    \item There exists $a,b \in \R$ such that $|a - b| \leq 10\gamma_i$ and $ \int_{ a \leq \dot{x_t}{q} \leq b} w_{i, t}(q) dq \geq 1 -\gamma_i^{4d}$.
    \item $C_{i-1} > 100\left(\frac{d^4(i-1)}{\gamma_{i-1}^{10d}} + d^{25}(i-1)\right)$.
\end{itemize}
\State Let $y$ be the weighted median of $w_{i_t, t}$ in the direction $x_t$ i.e. $y$ satisfies
\[
 \int_{\dot{x_t}{q} < y} w_{i_t, t}(q) dq   =  \int_{\dot{x_t}{q} > y}w_{i_t, t}(q) dq = \frac{1}{2}.
\]
\State Query $\hat{y} = y + \delta$ where $\delta$ is chosen uniformly at random from $[-2\beta_{i_t}, 2\beta_{i_t}]$.  
\State Update weights:
\For {$i$ in $\{i_t, i_t + 1\}$}
\State For each $q \in \B(0,1)$, if $q$ violates the feedback then $w_{i, t+1}(q) = \left(1 - \eta \right)w_{i, t}(q)$.
\State Normalize $w_i$ so that $\int_{\B(0,1)}w_{i, t+1}(p) dp = 1$.
\EndFor
\State $C_{i_t} \leftarrow C_{i_t} + 1$.
\EndFor
\end{algorithmic}
\end{algorithm}

\begin{theorem}\label{noisycontextualsearch}
Algorithm \ref{alg:noisy_contextual_search_alg} incurs $O(\poly(d))$ expected total loss for the problem of noisy linear contextual search.  
\end{theorem}

The proof of Theorem \ref{noisycontextualsearch} is structured roughly as follows. Let $L_i$ be the (expected) loss sustained at scale $i$. We wish to show that $\sum_{i} L_i = \poly(d)$. To bound $L_i$, we'll start by roughly following the analysis in Theorem \ref{noisygeneral}. Specifically, we'll look at the total weight (according to $w_i$) of a tiny ball surrounding the true point $q_0$. Let the weight of this ball at time $t$ be $W_{i, t}$. We'll again show that $1/W_{i, t}$ when suitably scaled is a super-martingale: it decreases in expectation by a large amount whenever our guess is far from accurate and cannot increase very much in expectation even if our guess is close to accurate. Since $1/W_{i,t}$ cannot decrease below $1$, this lets us upper bound the number of rounds where we are far from accurate. 

Now, even when we are far from accurate, we know that since we are in scale $i$, almost all of the mass of $w_{i}$ is concentrated on some thin strip in direction $x_t$. If the true point $q_0$ is located in or near this strip, this lets us bound the loss each round when we are far from accurate (since the median will lie in this strip). So it suffices to show that if a weight function $w_i$ concentrates on some thin strip, then with high probability, the true point $q_0$ lies close to this strip.

To prove this, we again look at the weight of a small ball $\B_{\alpha} = B(q_0, \alpha)$ with radius $\alpha$ around $q_0$ (see left side of Figure \ref{fig:strip_intersection}). If we know that the weight on some strip is at least some threshold $\tau$, then if $w_{i,t}(\B_{\alpha}) + \tau > 1$, we know that the ball and strip intersect, and therefore $q_0$ is at most distance $\alpha$ away from this strip. It thus suffices to show that $w_{i,t}(\B_{\alpha}) > 1 - \tau$ with high probability.

\begin{figure}[h]
\centering
\begin{subfigure}{0.40\textwidth}
\begin{tikzpicture}[scale=1.2]
    
    \begin{scope}
      \clip (-.2+1,-2) rectangle (.2+1,2);
      \fill[blue!20!white] (0,0) circle (2);
    \end{scope}
    \draw (0,0) circle (2);

    \draw[dotted,fill=yellow!70!white, opacity=.6, line width=1pt] (.65,-.4) circle (.35);
\node [shape=circle, fill=black,inner sep=1.5pt,label=left:$q_0$] (X1) at (.65,-.4) {};
\node at (.55,-.95) {$\B_\alpha$};
\end{tikzpicture}
\end{subfigure}
\begin{subfigure}{0.40\textwidth}
\begin{tikzpicture}[scale=1.2]

     \begin{scope}
       \clip (-3,1.5)--(2,0)--(2,-2)--(-3,-2)--cycle;
       \fill[red!10!white] (0,0) circle (2);
     \end{scope}
    \draw (0,0) circle (2);

    \draw[dotted,fill=yellow!70!white, opacity=.6, line width=1pt] (.65,-.4) circle (.35);
\node [shape=circle, fill=black,inner sep=1.5pt,label=left:$q_0$] (X1) at (.65,-.4) {};
\node [shape=circle, fill=black,inner sep=1.5pt,label=right:$q_1$] (X1) at (.7,-.2) {};
\node [shape=circle, fill=black,inner sep=1.5pt,label=left:$q_2$] (X1) at (1,1) {};
\draw (.65,-.4)--(1,1);
\end{tikzpicture}
\end{subfigure}
\caption{}
\label{fig:strip_intersection}
\end{figure}

Now, we will choose $\alpha$ and $\tau$ large enough so that this inequality is satisfied at time $t=0$. We therefore only need to show that this is still true with high probability for all times $t$. Intuitively, this should be true -- the amount of weight on a ball centered at the true point $q_0$ should only increase as time goes on and we get more feedback (the feedback is noisy, so we might occasionally decrease the weight of this ball, but overall the increases should drown out the decreases). Proving this formally, however, is technically challenging and where we need to use the Euclidean geometry specific to linear contextual search. 

To show this, we use the following lemma. Choose two points $q_1$ and $q_2$ on a line through $q_0$ so that $q_1$ lies between $q_0$ and $q_2$. We claim that with high probability (for all times $t$), $w_{i, t}(q_1) \geq \kappa \cdot w_{i, t}(q_2)$ for some constant $\kappa$. To show this, observe that there is no half space which contains both $q_0$ and $q_2$ but not $q_1$. This means that the only way $w_{i,t}(q_2)$ can increase relative to $w_{i,t}(q_1)$ is if a guess separates $q_2$ from $q_1$ and if the feedback on this guess is noisy (right side of Figure \ref{fig:strip_intersection}). This occurs with probability $p < 1/2$ and is unlikelier than the alternative (which increases the weight of $q_1$ relative to $q_2$). We can thus bound the ratio of $w_{i, t}(q_1)/w_{i, t}(q_2)$ from below with high probability over all rounds.

If we could union bound over all points in $\B_{\alpha}$ we would be done (this inequality allows us to relate the weight of all the points outside $\B_{\alpha}$ to the weight of points inside $\B_{\alpha}$). Unfortunately there are infinitely many points inside $\B_{\alpha}$ so we cannot apply a naive union bound. Luckily, we can show that nearby points are very likely to have similar weights: the only way the relative weight of two nearby points $q$ and $q'$ changes is if we guess a hyperplane separating $q$ and $q'$ -- and since we add a perturbation to our guess every round, we can bound the probability of this happening. This allows us to repeat the previous geometric argument with $\epsilon$-nets instead of single points, which completes the proof.\\

\paragraph{Notation.} Below, we will use $q_0$ to denote the hidden point.  We let $\B(0,1)$ denote the unit ball and in general $\B(q,r)$ to denote the ball of radius $r$ centered at $q$. We will use $w_{i,t}$ to denote the weight function $w_i$ at round $t$.  For a set $S \subset \B(0,1)$, we use the notation
\[
w_{i,t}(S)= \int_{S}w_{i,t}(q) dq.
\]

Let $\alpha_i = \frac{1}{2^{10^4d^2i}}$.  Define the set $S_{\alpha_i}$ to be the $\alpha_i$-net consisting of all points in the unit ball whose coordinates are integer multiples of $\frac{\alpha_i}{d}$.  Note that $|S_{\alpha_i}| \leq \left(\frac{2d}{\alpha_i}\right)^d$.  For all $i$, let $\Gamma_i = \B(q_0, \alpha_{i}) \cap \B(0,1)$ be the ball of radius $\alpha_i$ centered at $q_0$. For simplicity, throughout this proof we will assume that the feedback noise is fixed at $p = 1/3$ (it is straightforward to adapt this proof for any other $p < 1/2$; doing so only affects the constant factor of the loss bound). 

\subsubsection*{Step 1: Understanding $1/w_{i,t}$}

As in the analysis of Algorithm \ref{alg:general_noise}, we begin by understanding how the reciprocal of our weight function $1/w_{i, t}(p)$ evolves over time. This will allow us to construct various helpful super-martingales (for example, allowing us to bound the number of rounds we spend in each scale). 

The following claim relates how $1/w_{i,t}(q_1)$ changes when $q_1$ and $q_0$ are on the same side of the hyperplane $\dot{x_t}{q} = \hat{y}$ (i.e. is more likely to be consistent with feedback).

\begin{claim}\label{sameside}
Consider a round $t$.  Say the adversary picks direction $x_t$ and the algorithm queries $\hat{y}$.  Let $q_1$ be a point such that $q_1$ and $q_0$ are on the same side of the hyperplane $\dot{x_t}{q} = \hat{y}$.  Let 
\[
X = \int_{\mathsf{sign}(\dot{x_t}{q} - \hat{y}) \neq  \mathsf{sign}(\dot{x_t}{q_0} - \hat{y})} w_{i,t}(q)dq
\]
Then for $\eta \leq 1/4$ we have the following after applying weight updates
\[
\E\left[\frac{1}{w_{i,t+1}(q_1)} \right] \leq \left(1 - \frac{1}{10}\eta X \right)\frac{1}{w_{i,t}(q_1)}
\]
where the expectation is over the randomness in the feedback.  In particular, we always have
\[
\E\left[\frac{1}{w_{i,t+1}(q_1)} \right] \leq \frac{1}{w_{i,t}(q_1)}.
\]
\end{claim}
\begin{proof}

Recall that points that violate feedback have their weight multiplied by $(1-\eta)$ (and then the distribution is renormalized). With probability $1-p = 2/3$ (when the feedback is not flipped), we thus have that

$$w_{i, t+1}(q_1) = \frac{w_{i, t}(q_1)}{(1 - X) + (1-\eta) X}.$$

Likewise, with probability $p = 1/3$ (when the feedback is flipped), we have that

$$w_{i, t+1}(q_1) = \frac{(1-\eta) w_{i, t}(q_1)}{(1-\eta)(1 - X) + X}.$$

Taking expectations over the feedback, we therefore have that
\begin{eqnarray*}
\E\left[\frac{1}{w_{i,t+1}(q_1)} \right] &=& \left( \frac{2}{3} \cdot\left(1 - \eta X\right) + \frac{1}{3} \cdot \left(1 + \frac{\eta}{1-\eta}X \right) \right)\frac{1}{w_{i,t}(q_1)} \\
&=&  \left( 1 - \eta X \left(\frac{2}{3} - \frac{1}{3 (1-\eta)}\right) \right)\frac{1}{w_{i,t}(q_1)} \\
&\leq & \left( 1 - \frac{1}{10}\eta X \right)\frac{1}{w_{i,t}(q_1)}
\end{eqnarray*}
\end{proof}

Points very close to $q_0$ are likely to be on the same side of the hyperplane as $q_0$, allowing us to apply Claim \ref{sameside}.

\begin{claim}\label{closest}
Let $q_1 \in \B(0,1)$ such that $\norm{q_1 - q_0}
\leq \alpha_i$.  Then, in expectation both over the randomness in the feedback and the algorithm, 
\[
\E\left[\frac{1}{w_{i,t+1}(q_1)} \right] \leq \frac{1}{ 1 - \frac{\alpha_i}{\beta_i}} \frac{1}{w_{i,t}(q_1)}
\]
regardless of the direction $x_t$ that the adversary chooses.
\end{claim}
\begin{proof}
Since $\norm{q_1 - q_0} \leq \alpha_i$, and since $\hat{y}$ is chosen by adding a uniform $\beta_i$ random variable to $y$, the probability that $q_1$ and $q_0$ are on opposite sides of the plane $\dot{x_t}{q} = \hat{y}$ is at most $\frac{\alpha_i}{\beta_i}$.  Combining this with Claim \ref{sameside} gives us the desired result.
\end{proof}

We now use Claims \ref{closest} and \ref{sameside} to understand how $1/w_{i,t}(\Gamma_i)$ changes over time (generalizing from single points to small balls). This first claim bounds the decrease in $1/w_{i_t+1, t+1}(\Gamma_{i_t+1})$ when our guess is close to accurate.

\begin{claim}\label{weightincrease1}
Assume $ |y - \dot{x_t}{q_0}| \leq \beta_{i_t}$.  Then, in expectation both over the randomness in feedback and in our algorithm,
\[
\E\left[\frac{1}{w_{i_t+1,t+1}(\Gamma_{i_t + 1})} \right] \leq \left(1 - \gamma_{i_t}^{10d} \right)\frac{1}{w_{i_t+1,t}(\Gamma_{i_t + 1})}.
\]
\end{claim}
\begin{proof}
Recall that $y$ is the median of $w_{i_t, t}$ in direction $x_t$ as computed by our algorithm. Now, define the two quantities
\[
X^{-} = \int_{\dot{x_t}{q} \leq y - 2\beta_{i_t}} w_{i_t + 1, t}(q)dq 
\quad \text{and} \quad
X^{+} = \int_{\dot{x_t}{q} \geq y +  2\beta_{i_t}} w_{i_t + 1, t}(q)dq.
\]

These quantities represent the mass of $w_{i_t + 1}$ above and below the strip of width $2\beta_{i_t}$ around the median. Note that by the maximality of $i_t$, either $X^- \geq \gamma_{i_t + 1}^{4d}/2$ or $X^+ \geq \gamma_{i_t + 1}^{4d}/2$; if not, then there exists a strip of width $2\beta_{i_t} \leq 10\gamma_{i_t+1}$ containing at least $1 - \gamma_{i_t+1}^4d$. Without loss of generality, assume $X^- \geq \gamma_{i_t + 1}^{4d}/2$. 

Now, recall that $\hat{y}$ is chosen uniformly in the interval $[y - 2\beta_{i_t}, y + 2\beta_{i_t}]$. We will divide the expectation in the theorem statement into three cases, based on where $\hat{y}$ lies.

\begin{itemize}
    \item \textbf{Case 1}: $\hat{y} \in [y - 2\beta_{i_t}, y - \beta_{i_t} - \alpha_{i_t+1}]$.
    
    This case occurs with probability $\frac{1}{4} - \frac{\alpha_{i_t+1}}{2\beta_{i_t}}$. Note that since $|y - \dot{x_t}{q_0}| \leq \beta_{i_t}$, in this case we also have that $\hat{y} \leq \dot{x_t}{q} - \alpha_{i_t+1}$. Therefore in this case we know that the ball $\Gamma_{i_t + 1}$ lies entirely to the left of the hyperplane $\dot{x_t}{q} = \hat{y}$. By applying Claim \ref{sameside}, we know that, conditioned on being in this case,
    
    \begin{eqnarray*}
    \E\left[\frac{1}{w_{i_t+1, t+1}(\Gamma_{i_t + 1})}\right] &\leq & \left(1 - 0.1X^{-}\eta\right)\frac{1}{w_{i_t+1, t+1}(\Gamma_{i_t + 1})} \\
    &\leq & \left(1 - 0.05\gamma_{i_t + 1}^{4d}\eta\right)\frac{1}{w_{i_t+1, t+1}(\Gamma_{i_t + 1})}.
    \end{eqnarray*}
    
    \item \textbf{Case 2}: $\hat{y} \in [\dot{x_t}{q_0} - \alpha_{i_t+1}, \dot{x_t}{q_0} + \alpha_{i_t + 1}]$.
    
    This case covers the $\hat{y}$ where the hyperplane $\dot{x_t}{q} = \hat{y}$ intersects the ball $\Gamma_{i_t + 1}$. For $\hat{y}$ in this case, we pessimistically bound the change in weight via 
    
    $$\frac{1}{w_{i_t+1, t+1}(\Gamma_{i_t + 1})} \leq \frac{1}{1 - \eta}\frac{1}{w_{i_t+1, t+1}(\Gamma_{i_t + 1})}.$$
    
    Luckily, this case occurs rarely, with probability $\frac{\alpha_{i_t + 1}}{\beta_{i_t}}$.
    
    \item \textbf{Case 3}: remainder of the interval $[y - 2\beta_{i_t}, y + 2\beta_{i_t}]$.
    
    Since case 1 and case 2 together cover at least $1/4$ of the interval, this case occurs with probability at most $3/4$. In this case the ball $\Gamma_{i_t+1}$ does not intersect the hyperplane (since all such $\hat{y}$ are covered by case 2). We can therefore apply (the weaker variant of) Claim \ref{sameside} to show that, conditioned on being in this case,

    $$\E\left[\frac{1}{w_{i_t+1, t+1}(\Gamma_{i_t + 1})}\right] \leq \frac{1}{w_{i_t+1, t+1}(\Gamma_{i_t + 1})}. $$
\end{itemize}

Combining these three cases, we have that

\begin{eqnarray*}
\E\left[\frac{1}{w_{i_t+1,t+1}(\Gamma_{i_t + 1})} \right] &\leq& \left(\left( \frac{1}{4} - \frac{\alpha_{i_t+1}}{2\beta_{i_t}}\right)\left( 1 - 0.05\gamma_{i_t + 1}^{4d}\eta \right) + \frac{\alpha_{i_t+1}}{\beta_{i_t}}\frac{1}{\left(1 - \eta \right)} + \frac{3}{4}\right)\frac{1}{w_{i_t+1,t}(\Gamma_{i_t + 1})}  \\ 
&\leq & \left(1 - \frac{1}{80}\gamma_{i_t}^{4d}\eta + 2\frac{\alpha_{i_{t+1}}}{\beta_{i_t}}\right)\frac{1}{w_{i_t+1,t}(\Gamma_{i_t + 1})} \leq  \frac{\left(1 - \gamma_{i_t}^{10d} \right)}{w_{i_t+1,t}(\Gamma_{i_t + 1})}.
\end{eqnarray*}

\end{proof}

When our guess is far from accurate, we can instead (more strongly) bound the decrease in $1/w_{i_t, t}(\Gamma_t)$. 

\begin{claim}\label{weightincrease2}
Assume $|y - \dot{x_t}{q_0}| > \beta_{i_t}$.  Then, in expectation both over the randomness in feedback and in our algorithm,
\[
\E\left[ \frac{1}{w_{i_t, t+1}(\Gamma_{i_t})}\right] \leq \left(1 - \frac{1}{d^{21}} \right)\frac{1}{w_{i_t,t}(\Gamma_{i_t})}
\]
\end{claim}
\begin{proof}
We essentially repeat the logic from the proof of Claim \ref{weightincrease1}, with the change that we can more strongly lower bound $X^{-}$. Without loss of generality, assume that $y < \dot{x_t}{q_0} - \beta_{i_t}$. Define

\[
X^{-} = \int_{\dot{x_t}{q} \leq \hat{y}}w_{i_t,t}(q)dq.
\]
\noindent
Note that since $y$ is the weighted median of $w_{i_t, t}$ in the direction $x_t$, $X^{-} \geq 1/2$. 

Now, we again have three cases. To begin, with probability $\frac{1}{4} - \frac{\alpha_{i_t}}{\beta_{i_t}}$, $\hat{y}$ lies in the interval $[y, y + \beta_{i_t} - \alpha_{i_t}]$. Since $y + \beta_{i_t} < \dot{x_t}{q_0}$, the hyperplane $\dot{x_t}{q} = \hat{y}$ does not intersect $\Gamma_{i_t}$, and therefore we can apply Claim \ref{sameside} to show that

\begin{eqnarray*}
\E\left[\frac{1}{w_{i_t+1, t+1}(\Gamma_{i_t + 1})}\right] &\leq & \left(1 - 0.1X^{-}\eta\right)\frac{1}{w_{i_t+1, t+1}(\Gamma_{i_t + 1})} \\
&\leq & \left(1 - 0.05\eta \right)\frac{1}{w_{i_t+1, t+1}(\Gamma_{i_t + 1})}.
\end{eqnarray*}

Likewise, the probability that the hyperplane $\dot{x_t}{q} = \hat{y}$ intersects $\Gamma_{i_t}$ is at most $\frac{\alpha_{i_t}}{\beta_{i_t}}$ (in which case we can pessimistically bound the decrease in weight as in the proof of Claim \ref{weightincrease1}), and with the remaining $3/4$ probability the hyperplane $\dot{x_t}{q} = \hat{y}$ does not intersect $\Gamma_{i_t}$, and we can apply the weaker variant of Claim \ref{sameside}. Combining these observations, we get that

\begin{eqnarray*}
\E\left[ \frac{1}{w_{i_t, t+1}(\Gamma_{i_t})}\right] &\leq& \left(\left( \frac{1}{4} - \frac{\alpha_{i_t}}{\beta_{i_t}} \right) \left(1 - 0.05\eta \right) + \frac{\alpha_{i_t}}{\beta_{i_t}}\frac{1}{\left( 1 - \eta \right)} + \frac{3}{4}\right) \frac{1}{w_{i_t,t}(\Gamma_{i_t})} \\ 
&\leq& \left(1 - \frac{1}{d^{21}} \right)\frac{1}{w_{i_t,t}(\Gamma_{i_t})}.
\end{eqnarray*}

\end{proof}

\subsubsection*{Step 2: Bounding the number of rounds}

In this step we will bound the total number of rounds in each scale $i$. Specifically, our algorithm ensures that we move onto the next scale once either the weight concentrates on a strip or once $C_i$ grows large enough - we will show with high probability that this is always due to the weight concentrating on a small strip. 

Let $A_i$ be the number of rounds $t$ such that $i_t = i$ and $|y - \dot{x_t}{q_0}| \leq \beta_i$ (i.e. the number of rounds where we are ``accurate'').  Let $B_i$ be the number of rounds $t$ such that $i_t = i$ and $|y - \dot{x_t}{q_0}| > \beta_i$ (i.e. the number of rounds where we are ``inaccurate'').  Note that $A_i + B_i = C_i$. Also, recall that our algorithm ensures that $C_{i} \leq 100(\frac{d^4i}{\gamma_i^{10d}} + d^{25}i) + 1$ for all rounds $t$.

We first show that with high probability, $B_i$ will be no larger than $O(\poly(d)i)$.
\begin{claim}\label{farquerybound}
For any constant $c > 0$, with probability at least
$1 - {2^{-d^4i c}}$
we have throughout all rounds that
$$B_i \leq 100d^{25}i (1 + c)$$
\end{claim}
\begin{proof}
We will construct a sequence $Z_t$ so that $\frac{Z_t}{w_{i,t}(\Gamma_i)}$ is a super-martingale. Consider the sequence $Z_t$ defined as follows.
\begin{itemize}
    \item $Z_1 = \left(\frac{\alpha_i}{2}\right)^d$.
    \item If $i_t \not\in \{i, i-1\}$ then $Z_{t+1} = Z_t$.
    \item If $i_t = i$ and $|y - \dot{x_t}{q}| \leq \beta_i$ or $i_t = i-1$ then $Z_{t+1} = \left(1 - \frac{\alpha_i}{\beta_i}\right)Z_t$.
    \item If $i_t = i$ and $|y - \dot{x_t}{q}| > \beta_i$ then $Z_{t+1} = \left(1 + \frac{1}{d^{21}}\right)Z_t$.
\end{itemize}
Consider the ratio $Y_t = \frac{Z_t}{w_{i,t}(\Gamma_i)}$.  Note that Claim \ref{closest} and Claim \ref{weightincrease2} imply that
\[
\E[Y_{t+1} | Y_t] \leq Y_t,
\]
so $Y_t$ is a super-martingale.

Now, note that $Y_1 \leq 1$ since 
\[
w_{i,1}(\Gamma_i) = \frac{\Vol(\B(q_0, \alpha_i) \cap \B(0,1))}{\Vol(\B(0,1))} \geq \frac{1}{2^d}\cdot \frac{\Vol(\B(q_0, \alpha_i))}{\Vol(\B(0,1))} = \frac{\alpha_i^d}{2^d}
\]
(Here we have used the fact that $\Vol(\B(q_0, \alpha_i) \cap \B(0,1))$ must contain a ball of radius $\alpha_i / 2$.) Since $Y_t$ is a non-negative super-martingale, by Doob's martingale inequality it holds that for any constant $M$
\[
\Pr[\exists t| Y_t \geq M] \leq \frac{1}{M}.
\]
However note that if there exists a round where $B_i \geq 100d^{25}i (1 + c)$, then for $t$ sufficiently large

\begin{eqnarray*}
Z_t &\geq& Z_1\left(1 + \frac{1}{d^{21}} \right)^{B_i} \left( 1 - \frac{\alpha_i}{\beta_i}\right)^{C_i + C_{i-1}}\\ 
&\geq& \frac{\alpha_i^d}{2^d} \left(1 + \frac{1}{d^{21}} \right)^{100d^{25}i (1 + c)} \left(1 - \frac{\alpha_i}{\beta_i} \right)^{1000\left(\frac{d^4i}{\gamma_i^{10d}} + d^{25}i\right) } \\
&\geq& \frac{\alpha_i^d}{2^{d}} \cdot  2^{d^4i(1+c)} \\
&\geq& 2^{d^4ic}.
\end{eqnarray*}

(Here in the last inequality we have used the fact that $2^{d^4i} \geq (2/\alpha_i)^d$). Since $w_{i,t}(\Gamma_i) \leq 1$ for all $t$, this implies that $Y_t \geq 2^{d^4ic}$, which immediately implies the desired claim.
\end{proof}

Recall that in our algorithm, we check the following two conditions for determining the scale $i$ we use for the current query:
\begin{itemize}
    \item There exists $a, b \in \R$ such that $|a - b| \leq 10\gamma_i$ and $\int_{a \leq \dot{x_t}{q} \leq b} w_i(q) dq \geq 1 -\gamma_i^{4d}$.
    \item $C_{i-1} > 100(\frac{d^4(i-1)}{\gamma_{i-1}^{10d}} + d^{25}(i-1))$.
\end{itemize}

We now show that with high probability, only the first condition is ever relevant.

\begin{claim}\label{counterbound}
For an index $i$, the probability that we ever have
$$C_{i} \geq 100\left(\frac{d^4i}{\gamma_i^{10d}} + d^{25}i\right)$$ is at most
${2^{-d^4i}}$.
\end{claim}
\begin{proof}
Again, we will construct a sequence $Z_t$ so that $\frac{Z_t}{w_{i+1,t}(\Gamma_{i+1})}$ is a super-martingale. Consider the sequence $Z_t$ defined as follows.
\begin{itemize}
    \item $Z_1 = \frac{\alpha_{i+1}^d}{2^d}$
    \item If $i_t \not\in \{i, i+1\}$ then $Z_{t+1} = Z_t$.
    \item If $i_t = i$ and $|y - \dot{x_t}{q_0}| > \beta_i$ or $i_t = i+1$ then $Z_{t+1} = \left(1 - \frac{\alpha_{i+1}}{\beta_{i+1}} \right)Z_t$
    \item If $i_t = i$ and $|y - \dot{x_t}{q_0}| \leq \beta_i$ then $Z_{t+1} = \left( 1 + \gamma_i^{10d} \right) Z_t$
\end{itemize}
Consider the ratio $Y_t = \frac{Z_t}{w_{i+1,t}(\Gamma_{i+1})}$. Similarly as in the proof of Claim \ref{farquerybound}, $Y_1 \leq 1$. Note that Claim \ref{closest} and Claim \ref{weightincrease1} imply that 
\[
\E[Y_{t+1}] \leq Y_t
\]
\noindent
so $Y_t$ is a super-martingale.

Now assume that $C_{i} \geq 100\left(\frac{d^4i}{\gamma_i^{10d}} + d^{25}i\right)$. By the constraints of our algorithm, we are guaranteed that
\[
C_{i+1} \leq 100\left(\frac{d^4(i+1)}{\gamma_{i+1}^{10d}} + d^{25}(i+1)\right) + 1\leq \frac{1000d^{25}i}{\gamma_{i+1}^{10d}}
\]
Also by Claim \ref{farquerybound}, with probability at least $1 - \frac{1}{2^{10d^4i}}$, $B_i \leq 1100d^{25}i$ over all rounds $t$.  This implies that eventually
\[
A_i = C_i - B_i \geq 99\left(\frac{d^4i}{\gamma_i^{10d}} \right) 
\]
Thus, for sufficiently large $t$, we have that
\[
Z_t \geq \frac{\alpha_{i+1}^d}{2^d} \left(1 + \gamma_i^{10d} \right)^{99\left(\frac{d^4i}{\gamma_i^{10d}} \right) } \left(1 - \frac{\alpha_{i+1}}{\beta_{i+1}} \right)^{C_{i+1} + B_i} \geq 2^{2d^4i}
\]
However note $Y_1 \leq 1$ and $Y_t$ is a supermartingale.  Also, $w_{i+1,t}(\Gamma_{i+1}) \leq 1$ for all rounds $t$.  Thus, by Doob's martingale inequality, the probability we ever have $C_{i} \geq 100\left(\frac{d^4i}{\gamma_i^{10d}} + d^{25}i\right)$ is at most
\[
\frac{1}{2^{10d^4i}} + \frac{1}{2^{2d^4i}} \leq \frac{1}{2^{d^4i}}
\]
(the first term is from the probability that at some point $B_i \geq 1100d^{25}i$).
\end{proof}

\subsubsection*{Step 3: Proving $w_i$ concentrates near $q_0$}

We now aim to show that with high probability, if the weight function $w_i$ is concentrated on a thin strip, this strip must be close to the true point $q_0$ (this is necessary to bound the total regret we incur each round in scale $i$). To do this, we will argue that we can ``round" points to the $\alpha_i$-net $S_{\alpha_i}$ without significantly affecting their weight.  We will then rely on the geometric observation mentioned earlier: that for points $q_1,q_2 \in S_{\alpha_i}$ for some $i$ such that $q_0$, $q_1$, and $q_2$ are nearly collinear, we can relate the weights $w_{i,t}(q_1)$ and $w_{i,t}(q_2)$. We begin by relating the weights of collinear points. 

\begin{claim}\label{collinear}
Fix an index $i$. If $q_0$, $q_1$, and $q_2$ are collinear in that order, then with probability at least
$1 - {2^{-d^{10}i}}$
we have that for all rounds $t$
$$w_{i, t}(q_1) \geq \gamma_i^d w_{i,t}(q_2)$$
\end{claim}
\begin{proof}
Consider a time step $t$ where $i_t = i$ or $i_t = i-1$.  We say a point is on the ``good" side of the hyperplane $\dot{x_t}{q} = \hat{y}$ if it is on the same side as $q_0$.  Otherwise we say the point is on the ``bad" side.  Note for $q_1, q_2$ satisfying the conditions of the claim, one of the following statements must be true:
\begin{itemize}
    \item \textbf{Case 1}: $q_1, q_2$ are on the same side of the hyperplane $\dot{x_t}{q} = \hat{y}$.
    \item \textbf{Case 2}: $q_1$ is on the good side of the hyperplane and $q_2$ is on the bad side of the hyperplane.
\end{itemize}

We will now consider the quantity $R_t = \left( \frac{w_{i,t}(q_2)}{w_{i,t}(q_1)} \right)^{d^9}$. Note that in Case 1, then $R_{t+1} = R_t$ (the relative weights remain unchanged if both $q_1$ and $q_2$ are on the same side of the hyperplane). In Case 2, 

\begin{eqnarray*}
\E[R_{t+1}] &=& \E\left[ \left( \frac{w_{i,t + 1}(q_2)}{w_{i,t + 1}(q_1)} \right)^{d^9}\right] = \left(\frac{2}{3} \left(1 - \eta \right)^{d^9} + \frac{1}{3} \left(\frac{1}{1 - \eta} \right)^{d^9}\right)\left( \frac{w_{i,t}(q_2)}{w_{i,t}(q_1)} \right)^{d^9} 
\leq  \left( \frac{w_{i,t}(q_2)}{w_{i,t}(q_1)} \right)^{d^9}
\end{eqnarray*}

Here the last inequality follows from the fact that $2x/3 + 1/(3x) \leq 1$ for all $x \in [1/2, 1]$, and $\eta = \Theta(d^{-10})$ so $(1 - \eta)^{d^9} = 1 - o(1) \in [1/2, 1]$. Note that this implies that $R_{t}$ is a non-negative super-martingale (with $R_1 = 1$). 

Now, if $w_{i, t}(q_1) < \gamma_i^d w_{i, t}(q_2)$, this would mean that

$$R_t = \left(\frac{w_{i, t}(q_2)}{w_{i, t}(q_1)}\right)^{d^9} > \gamma_i^{-d^{10}}.$$

By Doob's martingale inequality, the probability that this ever happens is at most $\gamma_i^{d^{10}} \leq 2^{-d^{10}i}$, as desired.
\end{proof}

We next relate the weights of nearby points $q_1$ and $q_2$. If $q_1$ and $q_2$ are close together, then it is unlikely they are ever separated by a hyperplane, and their weights should be similar. The following claim captures this intuition. 

\begin{claim}\label{approximate}
Fix an index $i$. If $q_1$, and $q_2$ satisfy $\norm{q_1 - q_2} \leq 2\beta_{i}^{10}$, then with probability at least
$1 - {2^{-d^{10}i}}$
we have that for all rounds $t$
$$w_{i, t}(q_1) \geq \gamma_i w_{i,t}(q_2)$$
\end{claim}
\begin{proof}
We will bound the number of rounds $t$ such that $i_t = i$ or $i_t = i-1$ and the hyperplane $\dot{x_t}{q} = \hat{y}$ intersects the segment connecting $q_1$ and $q_2$.  We will show that with high probability, this quantity is at most $d^9i$. Note that since $w_{i, t}(q_1)/w_{i, t}(q_2)$ is unchanged when $q_1$ and $q_2$ both lie on the same side of the hyperplane, and decreases by at most a factor of $(1-\eta)$ when they lie on different sides, this will show that with high probability 

\begin{eqnarray*}
w_{i, t}(q_1) &\geq& (1 - \eta)^{d^9i} w_{i, t}(q_2) 
\geq  (1 - 2d^{-10})^{d^9i} w_{i, t}(q_2) 
\geq 2^{-i}w_{i, t}(q_2) 
= \gamma_i w_{i, t}(q_2),
\end{eqnarray*}
\noindent
as desired.

Now, for a fixed round $t$, note that the probability that the hyperplane $\dot{x_t}{q} = \hat{y}$ intersects the segment connecting $q_1$ and $q_2$ is at most $\frac{\norm{q_1 - q_2}}{2\beta_i} \leq \beta_i^9$. There are at most 
\[
C_i + C_{i-1} \leq \frac{1000d^4i}{\gamma_i^{10d}}
\]
indices $t$ for which $i_t = i$ or $i_t = i-1$.  The probability that at least $d^9i$ of the hyperplanes $\dot{x_t}{q} = \hat{y}$ intersect the segment connecting $q_{2}$ and $q_1$ is at most
\[
\binom{\frac{1000d^4i}{\gamma_i^{10d}}}{d^9i}\beta_i^{9d^9i} \leq \left(\frac{1000d^4i\beta_i^9}{\gamma_i^{10d}} \right)^{d^9i} \leq 2^{-d^{10}i},
\]

\noindent
which implies our desired result.

\end{proof}

Finally, we apply Claims \ref{collinear} and \ref{approximate} to bound the relative weights for all approximately collinear pairs of points in $S_{\alpha}$. 

\begin{claim}\label{weightratio}
Fix an index $i$.  With probability at least $
1 - {2^{-d^{9}i}}$
the following claim holds: for all $q_1,q_2 \in S_{\alpha_i}$ such that 
\begin{itemize}
    \item The angle between the vectors $q_1 - q_0$ and $q_2 - q_0$ is at most $\beta_i^{10}$
    \item $\norm{q_1 - q_0} \leq \norm{q_2 - q_0}$
\end{itemize}
we have for all rounds $t$
\[
w_{i,t}(p_1) \geq \gamma_i^{2d}w_{i,t}(p_2)
\]
\end{claim}
\begin{proof}
Fix a pair of points $q_1,q_2 \in S_{\alpha_i}$ satisfying the conditions in the statement. Let $q_{\perp}$ be the foot of the perpendicular from $q_1$ to the segment connecting $q_2$ and $q_0$.   Note that
\[
\norm{q_{\perp} - q_1} \leq 2\beta_i^{10}
\]

Therefore, by the conditions of Claim \ref{approximate}, with probability at least $1 - 2^{-d^{10}i}$, for all rounds $t$, 

$$w_{i, t}(q_1) \geq \gamma_i w_{i, t}(q_\perp).$$

Note that $q_0$, $q_{\perp}$, and $q_2$ are collinear. Since $\norm{q_{\perp} - q_0} \leq \norm{q_1 - q_0} \leq \norm{q_2 - q_0}$, $q_{\perp}$ lies between $q_0$ and $q_2$ on this line. By Claim \ref{collinear}, this means that with probability at least $1 - 2^{-d^{10}i}$, for all rounds $t$, 
$$w_{i, t}(q_{\perp}) \geq \gamma_i^{d}w_{i, t}(q_2).$$

Combining these two claims, we know that with probability at least $1 - 2^{-d^{10}i + 1}$, for all rounds $t$,

$$w_{i, t}(q_1) \geq \gamma_i^{d+1}w_{i, t}(q_2).$$

This is for a specific pair of points in $S_{\alpha_i}$. Union bounding over all $|S_{\alpha_i}|^2$ pairs of points, we have that the theorem statement fails with probability at most

$$|S_{\alpha_i}|^2 2^{-d^{10}i + 1} = \left(\frac{2d}{\alpha_i}\right)^{2d}2^{-d^{10}i + 1} \leq 2^{-d^{9}i}.$$
\end{proof}

Finally, we prove that if $w_i$ concentrates on a strip, $q_0$ is within $\gamma_i$ of this strip. To show this, it suffices to show that the weight of $\B(q_0, \gamma_i)$ is large enough that it must intersect a sufficiently concentrated strip. We do this by using Claim \ref{weightratio} to relate the weight of points of $S_{\alpha_i}$ inside and outside $\B(q_0, \gamma_i)$.

\begin{claim}\label{stripweight}
Fix an index $i$.  With probability at least $1 - {2^{-d^8i}}$, the following statement holds for all $t$:
\begin{itemize}
    \item If there exist $a,b$ such that $|a - b| \leq 10\gamma_i$ and 
    \[
    \int_{a \leq \dot{x_t}{q} \leq b} w_{i,t}(q)dq \geq 1 -\gamma_i^{4d}
    \]
    then $a - \gamma_i \leq \dot{x_t}{q_0} \leq b + \gamma_i $.
\end{itemize}
\end{claim}
\begin{proof}
First for each point $q \in S_{\alpha_i}$, consider an axis-parallel box centered at that point with side length $\frac{\alpha_i}{d}$.  Now consider all rounds $t$ with $i_t = i$ or $i_t = i-1$ and all planes of the form $\dot{x_t}{q} = \hat{y}$ for these rounds.  We show that with high probability, all boxes intersect at most $d^9i$ of these planes.  Using essentially the same argument as in Claim \ref{approximate}, we find that this probability is at least
\[
1 - \left( \frac{2d}{\alpha_i}\right)^{d} \binom{\frac{1000d^4i}{\gamma_i^{10d}}}{d^9i}\beta_i^{9d^9i} \geq 1 - \frac{1}{2^{d^{10}i}}
\]
In particular, with at least $1 - \frac{1}{2^{d^{10}i}}$ probability, the following two inequalities hold:
\begin{equation}\label{smallball}
\int_{q \in \B(q_0, \gamma_i) \cap \B(0,1)}w_{i,t}(q)dq \geq \left(1 - \eta \right)^{d^9i}\left(\frac{\alpha_i}{d}\right)^d\sum_{\substack{q \in \B(q_0, \gamma_i - \alpha_i) \cap S_{\alpha_i}}}w_{i,t}(q).
\end{equation}
\begin{equation}\label{whole}
1 = \int_{q \in \B(0,1)}w_{i,t}(q)dq \leq \frac{1}{\left(1 - \eta \right)^{d^9i}}\left(\frac{\alpha_i}{d}\right)^d\sum_{q \in S_{\alpha_i}}w_{i,t}(q).
\end{equation}
In both of these inequalities we are using the fact that if at most $d^9i$ planes intersect any box, then the weights of any two points in the same box are within a factor of $\left(1 - \eta \right)^{d^9i}$.
\\\\
Next, consider the ball $\B(q_0, \gamma_i - \alpha_i)$.  Let $T_i = \{\B(q_0, \gamma_i- \alpha_i) \cap S_{\alpha_i} \}$. Consider the following two transformations: $f : S_{\alpha_i} \rightarrow \B(q_0, \gamma_i - \alpha_i)$, which sends a point $q$ to

\[
f(q) = \left(1 - \frac{\gamma_i}{2}\right)q_0 + \frac{\gamma_i}{2}q
\]
and the transformation $g : \B(q_0, \gamma_i - \alpha_i) \rightarrow T_i$, where $g(q)$ is the point obtained by rounding the coordinates of $q$ to the nearest integer multiple of $\frac{\alpha_i}{d}$ (note that $g(q) \in T_i$). If we consider the map $q \rightarrow q' = g(f(q))$ given by the above, the number of points $q \in S_{\alpha}$ that map to a fixed point $q' \in T_i$ is at most 
\[
\left(\frac{10}{\gamma_i}\right)^d.
\]
To see this, note that $g^{-1}(q)$ is an axis-parallel box with side-length $\frac{\alpha_i}{d}$, and thus $f^{-1}(g^{-1}(q))$ contains all the points in $S_{\alpha}$ contained within an axis aligned box with side-length $\frac{2\alpha_i}{d\gamma_i}$, which contains at least $(2/\gamma_i + 1)^d < (10/\gamma_i)^d$ points.

Now, note that $q$ and $q'=g(f(q))$ satisfy the conditions of Claim \ref{weightratio}.  Thus, by Claim \ref{weightratio}, with probability at least $1 - \frac{1}{2^{d^9i}}$
we have that
\[
\sum_{q \in \B(q_0, \gamma_i - \alpha_i) \cap S_{\alpha_i}}w_{i,t}(q) \geq \gamma_i^{2d}\left( \frac{\gamma_i}{10}\right)^d \sum_{q \in S_{\alpha_i}}w_{i,t}(q) = \frac{\gamma_i^{3d}}{10^d}\sum_{q \in S_{\alpha_i}}w_{i,t}(q)
\]
Combining the above with equations (\ref{smallball}) and (\ref{whole}), we conclude that with probability at least 
\[
1 -\frac{1}{2^{d^9i}}  - \frac{1}{2^{d^{10}i}} \geq 1 - \frac{1}{2^{d^8i}}
\]
we have
\[
\int_{q \in \B(q_0, \gamma_i)}w_{i,t}(q)dq \geq \left(1 - \frac{1}{2d^{10}} \right)^{2d^9i}\frac{\gamma_i^{3d}}{10^d} \geq \gamma_i^{4d}
\]
This implies that the ball $\B(q_0, \gamma_i)$ must intersect the strip $a \leq \dot{x_t}{q} \leq b $.  If this happens then the desired condition is clearly satisfied.
\end{proof}

\subsubsection*{Step 4: Completing the proof}

Finally, we can proceed to prove the main theorem.

\begin{proof}[Proof of Theorem \ref{noisycontextualsearch}]
First, for each $i$, let $L_i$ be the total loss at scale $i$. We will bound $\E[L_i]$. 

By Claim \ref{counterbound}, with probability at least $1 - \frac{1}{2^{d^4(i-1)}}$, $C_{i-1} \leq  100(\frac{d^4(i-1)}{\gamma_{i-1}^{10d}} + d^{25}(i-1))$ for all rounds.  If this is true, then the only time we query at level $i$, there must be some strip given by $a \leq \dot{x_t}{q} \leq b$ of width at most $10\gamma_i$ that contains $1 - \gamma_i^{4d}$ of the total weight of $w_i$.  Thus, by Claim \ref{stripweight}, with at least $ 1- \frac{1}{2^{d^4(i-1)}} -  \frac{1}{2^{d^8(i-1)}}$ probability, all queries at level $i$ incur loss at most $12\gamma_i + 2\beta_i \leq 14\gamma_i$.  Now, by using Claim \ref{farquerybound}, we can bound the expected total loss at level $i$ as 
\begin{eqnarray*}
\E[L_i] &\leq& \left( \frac{1}{2^{d^4(i-1)}} +  \frac{1}{2^{d^8(i-1)}}\right)C_i + 2\beta_iC_i + 14 \gamma_i \left( 100d^{25}i + \sum_{j=0}^{\infty}100d^{25}i\Pr\left[ B_i \geq 100d^{25}i (1 + j)  \right] \right) \\
&\leq& \left( \frac{1}{2^{d^4(i-1)}} +  \frac{1}{2^{d^8(i-1)}} + 2\beta_i \right) \cdot \left( 100\left(\frac{d^4i}{\gamma_i^{10d}} + d^{25}i \right) + 1 \right) + 1400d^{25}i\gamma_i\left( 1 + \sum_{j=0}^{\infty} \frac{1}{2^{d^4ij}} \right) \\
&\leq & 4\beta_i \cdot \left(200\frac{d^{25}i}{\gamma_i^{10d}}\right) + 2800d^{25}i\gamma_i \\
&=& 4\cdot2^{-100di}\cdot (200d^{25}i 2^{10di}) + 2800d^{25}i2^{-i} \\
&=& 800d^{25}i 2^{-90di} + 2800d^{25}i2^{-i}.
\end{eqnarray*}

It follows that

$$\sum_{i=1}^{\infty} \E[L_i] \leq \sum_{i=1}^{\infty} \left(800d^{25}i 2^{-90di} + 2800d^{25}i2^{-i}\right) = O(d^{25}) = O(\poly(d)).$$

\end{proof}

\begin{remark}
Naively, one can implement Algorithm \ref{alg:noisy_contextual_search_alg} with time complexity $T^{O(d)}$, via the observation that $T$ hyperplanes divide $\B(0, 1)$ into at most $O(T^{d})$ pieces, so we can simply compute this division and the weight of each distribution $w_i$ (we care about at most $T$ scales) on each component of this division.

It is an interesting open question if it is possible to implement Algorithm \ref{alg:noisy_contextual_search_alg} (or otherwise achieve $O(\poly(d))$ regret) with time complexity $\poly(d, T)$. To do so, it would suffice to be able to efficiently sample from the distributions $w_{i}$.
\end{remark}

\section{Tight loss bounds for full feedback}\label{sec:full_feedback}

We also study the problem where the learner has full feedback, i.e, after the prediction $y_t$ the feedback is the actual value of $f_0(x_t)$. We show that the optimal regret can be completely characterized (up to constant factors) by a continuous analogue of the Littlestone dimension.

For this section we don't require the assumption that $\Y$ is ordered, only that the loss function forms a valid metric (i.e. is symmetric and satisfies the triangle inequality). 

\subsection{Tree Dimension}
\begin{definition}
A $(\X, \Y)$-tree of cost $c$ is a labeled binary tree with the following properties
\begin{itemize}
    \item There is a root node and each interior node has two children
    \item Each interior node is labeled with a triple $(x,y_1,y_2)$ where $x \in \mathcal{X}$, $y_1,y_2 \in \mathcal{Y}$
    \item For each leaf, the sum of $\ell(y_1,y_2)$ over all nodes on the path from the root to the leaf is at least $c$.
\end{itemize}
\end{definition}

\begin{definition}
We say a $(\X, \Y)$-tree $\Tree$ is $\H$-satisfiable if we can label each leaf with some $f \in \H$ such that for each node $(x,y_1,y_2) \in \Tree$, all leaves of the left subtree satisfy $f(x) = y_1$ and all leaves of the right subtree satisfy $f(x) = y_2$.
\end{definition}

\begin{definition}[Tree dimension]
We define $\tau(\H)$, the tree dimension of $\H$, to be the maximum cost of a $(\X, \Y)$-tree that is $\H$-satisfiable.
\end{definition}

\begin{remark}
Note we can naturally extend the above definition to any subset $\H' \subset \H$.
\end{remark}

It is worth noting that covering dimension is ``more restrictive" than tree dimension in the sense that bounded covering dimension implies bounded tree dimension. 
\begin{theorem}\label{covertreedim}
Let $\H$ be a hypothesis class consisting of functions mapping $\X \rightarrow \Y$ and let $L$ be a loss function that defines a metric on $\Y$.  If $\Cdim(\H)$ is finite then 
\[
\tau(\H) \leq 6 \cdot \Cdim(\H)
\]   
\end{theorem}
Before proving the above we prove a few preliminary lemmas.
\begin{lemma}\label{completetree}
Let $\H$ be a hypothesis class and $\ell$ be a loss function that defines a metric on $\Y$.  Let $\Tree$ be an $\H$-satisfiable $(\X, \Y)$-tree where all leaves have depth $d$ and such that for each internal node $(x,y_1,y_2)$, $\ell(y_1,y_2) > 2^{-i}$.  Then
\[
d \leq (i+1) \cdot \Cdim(\H)
\]
\end{lemma}
\begin{proof}
Since the tree is $\H$-satisfiable, we can label the leaves with functions $f \in \H$.  Any two of these functions $f_1,f_2$ must satisfy $d_{\infty}(f_1,f_2) \geq 2^{-i}$ since there must be some internal node $(x,y_1,y_2)$ where $f_1(x) = y_1$ and $f_2(x) = y_2$.  Therefore, there are $2^d$ functions in $\H$ such that any two have $d_{\infty}$ distance bigger than $2^{-i}$.  This implies that 
\[
\left|N_{{2^{-(i+1)}}}(\H) \right| \geq 2^d
\]
Now by the definition of covering dimension, we conclude
$d \leq(i+1) \cdot \Cdim(\H)$.
\end{proof}

Given a rooted binary tree $\Tree$, we say a rooted binary tree $\Tree'$ is contained in $\Tree$ if all of the nodes of $\Tree'$ are nodes of $\Tree$ and the nodes of $\Tree'$ form a binary tree where each interior node has two children under the topology given by $\Tree$.
\begin{lemma}\label{colortree}
Consider a rooted binary tree $\Tree$ (where all interior nodes have exactly two children) and say its nodes are colored with colors $1,2, \dots , c$.  We say the colored tree satisfies property $(x_1, \dots , x_c)$ if for each $i \in [c]$, it does not contain a monochromatic complete binary tree of color $i$ and depth $x_i$.  If the coloring of $\Tree$ satisfies property $(x_1, \dots , x_c)$, there exists a leaf such that on the path from the root to the leaf, there are at most $x_i$ nodes of color $i$ for all $i \in [c]$.   
\end{lemma}
\begin{proof}
We prove the lemma by induction on $x_1 + \dots + x_c$.  The base cases are obvious.  Now say the root of $\Tree$ is colored with color $i$.  Clearly we must have $x_i > 0$.  Then either the left or right subtree of the root must satisfy property $(x_1, \dots , x_i - 1, \dots , x_c)$.  Using the inductive hypothesis, we get the desired.
\end{proof}

\begin{proof}[Proof of Theorem \ref{covertreedim}]
Assume for the sake of contradiction that $\tau(\H) > 6 \cdot \Cdim (\H)$.  Consider a $(\X, \Y)$-tree $\Tree$ that is $\H$-satisfiable and has cost larger than $6\cdot \Cdim(\H)$.  Note we can assume that there are no nodes in $\Tree$ where $\ell(y_1,y_2) = 0$ since otherwise, we can delete that node and keep only its left subtree.  Let $c$ be an integer such that for all nodes $(x,y_1,y_2)$, we have $\ell(y_1,y_2) > 2^{-c}$.
\\\\
Now color the internal nodes of $\Tree$ with $c$ colors $\{1,2, \dots c \}$ where a node $(x,y_1,y_2)$ is color $i$ if 
\[ 
\frac{1}{2^{i}} < \ell(y_1,y_2) \leq \frac{1}{2^{i-1}} 
\]
Note by Lemma \ref{completetree}, $\Tree$ does not contain any monochromatic, complete binary trees of color $i$ with depth at least $(i+1) \cdot \Cdim(\H)$.  By Lemma \ref{colortree}, this implies the total cost of $\Tree$ is at most
\[
\sum_{i=1}^c\frac{(i+1) \cdot \Cdim(\H)}{2^{i-1}} \leq 6 \cdot \Cdim(\H)
\]
which completes the proof.  
\end{proof}

However, we cannot hope for any sort of converse to Theorem \ref{covertreedim} as evidenced by the following example.  Let $\mathcal{X} = \mathcal{Y} = [0,1]$ and let $\ell(y_1,y_2) = |y_1 - y_2|$.  Let $H = \{1_{x = c}\mid c \in [0,1] \}$ be the set of all indicator functions of points in $[0,1]$.  $H$ has infinite covering dimension but its tree dimension is just $1$.

\subsection{Regret Bounds from Tree Dimension}

\begin{theorem}\label{full_feedback}
In the full feedback model there exists an algorithm with regret $O(\tau(\H))$.  Furthermore, no algorithm can guarantee less than $\tau(\H)/2$ regret.
\end{theorem}

First we prove that the algorithm below achieves the upper bound.

\begin{algorithm}[H]
\caption{{\sc Contextual Binary Search} }\label{algo:cbs}
\begin{algorithmic}
\For {$t$ in $1,2, \dots , T$}
\State Adversary picks $x_t$
\State Let $S_t$ be the set of hypotheses consistent with the feedback so far
\State For each $\epsilon$ define $A_{\epsilon, t} = \{y| y\in \mathcal{Y}, \tau\left(S_t \cap \{f| f(x_t) = y \}\right)\geq \tau\left(S_t\right) - \epsilon\}$
\State Choose $y_t \in A_{\epsilon, t}$ for the smallest $\epsilon$ such that $A_{\epsilon, t} \neq \emptyset$
\EndFor
\end{algorithmic}
\end{algorithm}

\begin{lemma}\label{rangebound}
For any $y_1,y_2 \in A_{\epsilon,t}$, $\ell(y_1,y_2) \leq \epsilon$.  
\end{lemma}
\begin{proof}
Assume for the sake of contradiction that this is false.  Then we can construct a $S_t$-satisfiable $(\mathcal{X}, \mathcal{Y}, s_i)$-tree with $(x_t,y_1,y_2)$ as its root node and cost bigger than $\tau(S_t)$.
\end{proof}

\begin{lemma} For each $t$, we have $$\tau(S_{t+1}) \leq \tau(S_t) - \ell(y_t, f_0(x_t))$$
where $\ell(y_t, f_0(x_t))$ is the loss of the algorithm in round $t$.
\end{lemma} 

\begin{proof}
Since $\epsilon$ is the smallest value such that $A_{\epsilon,t}$ is non-empty, then for every $y \in A_{\epsilon,t}$ we must have $\tau\left(S_t \cap \{f| f(x_t) = y \}\right) = \tau\left(S_t\right) - \epsilon$. It follows that if $\ell(y_t, f_0(x_t)) \leq \epsilon$ we are done.\\

Consider now the case where $L := \ell(y_t, f_0(x_t)) > \epsilon$. For this case, we want to argue that $f_0(x_t) \notin A_{L', t}$ for any $L' < L$, since after we get the feedback, we will update $S_{t+1} = \{f \in S_t; f(x_t) = f_0(x_t)\}$. Therefore $f_0(x_t) \notin A_{L', t}$ for all  $L' < L$ implies that: $\tau(S_{t+1}) \leq \tau(S_t) - L$.

Pick any $L'$ with $\epsilon < L' < L$. To see that $f_0(x_t) \notin A_{L', t}$, observe that $y_t \in A_{\epsilon, t} \subseteq A_{L',t}$. If it were the case that  $f_0(x_t) \in A_{L', t}$, we would have $L \leq L'$ by Lemma \ref{rangebound}, contradicting the fact that  $L' < L$.
\end{proof}

 \begin{proof}[Proof of Theorem \ref{full_feedback}]
The upper bound follows directly from the previous lemma. For the lower bound, consider an $\left( \X, \Y\right)$-tree with cost $\tau(\H)$ that is $\H$-satisfiable.  We can ensure that all leaves have the same depth $d$ (by adding nodes of the form $(x,y,y)$).  Now the adversary chooses a leaf uniformly at random.  If the sequence of nodes from the root to the leaf are
\[
(x_1, y_{11}, y_{12}), \dots , (x_d, y_{d1}, y_{d2})
\]
then the adversary presents the inputs $x_1, x_2, \dots  x_d$ in that order to the learner.  Since the loss function satisfies the triangle inequality, the expected loss of any learner is at least $\tau(\H)/2$ so we are done.
\end{proof}

\begin{remark}
Algorithm \ref{algo:cbs} assumes that the set $\{\epsilon; A_{\epsilon, t} \neq \emptyset\}$ has a minimum, which is always the case if the hypothesis class $\H$ is finite. For infinite $\H$, this minimum might not exist. In such a case, choose $\epsilon = \inf \{\epsilon; A_{\epsilon, t} \neq \emptyset\}$ and choose $y_t \in A_{\epsilon + g_t, t}$ for $g_t = 1/2^t$. Theorem \ref{full_feedback} can then be easily adapted to provide a bound of $\tau(\H) + \sum_t g_t \leq \tau(\H) + 1$.
\end{remark}

\subsection{Separating binary and full feedback}

With binary feedback we can no longer obtain loss bounds that depend only on tree dimension.  To see this, consider the following example:

Let $\H$ be the set of all functions $f:[n] \rightarrow \{0,1/n, ..... (n-1)/n\}$.  There are $n^n$ such functions.  Now for each, slightly perturb the outputs (i.e. $f(i) = j/n + \eps$) so that for every $f_1,f_2$ and $i$, $f_1(i) \neq f_2(i)$.  Let $\ell(y_1,y_2) = |y_1 - y_2|$.  The tree dimension of this class is $O(1)$.  However, clearly any algorithm must incur $\Omega(n)$ loss in expectation with binary feedback.


\bibliographystyle{plain}
\footnotesize
\bibliography{pricing}
\normalsize

\end{document}